\pdfoutput=1
\documentclass[amsmath,amssymb,aps,pre,showpacs,twocolumn,groupedaddress]{revtex4-1}

\usepackage{amsthm}
\usepackage{graphicx}
\usepackage[T1]{fontenc}
\usepackage{times}
\usepackage{mathtools}
\usepackage{tensor}
\usepackage{stmaryrd}
\usepackage{MnSymbol}
\usepackage{enumitem}

\theoremstyle{plain}
\newtheorem{theorem}{Theorem}
\newtheorem{lemma}{Lemma}
\newtheorem{corollary}{Corollary}

\newcommand{\ve}[1]{\mathbf{#1}} % Vectors
\newcommand{\p}{\mathrlap{'}} % A prime that takes no space

\DeclareMathOperator{\Prob}{\mathbb{P}}

\DeclareMathOperator{\Li}{Li}

\DeclareMathOperator{\diff}{d}

\newcommand{\inAB}[1]{{#1}\mathrlap{\kern-0.05em\smash{\raisebox{-0.5ex}{\rotatebox{20}{\ensuremath{\scriptscriptstyle\rcurvearrowright}}}}}}
\newcommand{\amountsandleft}[1]{#1} % Turn off
\newcommand{\amountsandright}[1]{#1} % Turn off
\newcommand{\link}[2]{{#1}{#2}} % Turn off
\newcommand{\specdiff}[1]{(#1)}
\newcommand{\funY}[2]{\tensor*{Y}{_{\amountsandright{#1}}^{\amountsandright{#2}}}}
\newcommand{\funZ}[4]{\tensor*{Z}{_{\link{#1}{#2}}^{\link{#3}{#4}}}}
\newcommand{\funL}[2]{\tensor*{L}{_{\amountsandright{#1}}^{\amountsandright{#2}}}}
\newcommand{\funH}{H}
\newcommand{\funA}[2]{\tensor*[_{\inAB{#2}}^{\amountsandleft{#1}}]{A}{}}
\newcommand{\funB}[2]{\tensor*[_{\inAB{#2}}^{\amountsandleft{#1}}]{B}{}}
\newcommand{\funC}[1]{\tensor*[_{\inAB{#1}}]{C}{}}
\newcommand{\funD}[1]{\tensor*[_{\inAB{#1}}]{D}{}}
\newcommand{\funF}{F}
\newcommand{\funG}{G}
\newcommand{\funPt}{\widetilde{P}}
\newcommand{\funQt}{\widetilde{Q}}
\newcommand{\funPts}{\widetilde{P}^*}
\newcommand{\funQts}{\widetilde{Q}^*}
\newcommand{\funU}{U}
\newcommand{\funV}{V}
\newcommand{\funYy}[2]{\funY{#1}{#2}(\ve{y})}
\newcommand{\funZz}[4]{\funZ{#1}{#2}{#3}{#4}(\ve{z})}
\newcommand{\funLz}[2]{\funL{#1}{#2}(\ve{z})}
\newcommand{\funHxyz}{\funH(w,x,\ve{y},\ve{z})}
\newcommand{\funAxyz}[2]{\funA{#1}{#2}(w,x,\ve{y},\ve{z})}
\newcommand{\funBxyz}[2]{\funB{#1}{#2}(w,x,\ve{y},\ve{z})}
\newcommand{\funZtop}[1]{\funZ{2}{2}{i'\!}{#1}}

\newcommand{\funBip}[1]{\funB{i'\!}{#1}}
\newcommand{\diffyH}[1]{\tensor*{h}{_{\specdiff{#1}}}}
\newcommand{\diffzH}[1]{\tensor*{\eta}{_{\specdiff{#1}}}}
\begin{document}

\title{A Bottom-Up Model of Self-Organized Criticality on Networks}

\author{Pierre-Andr\'e No\"el}
\email{noel.pierre.andre@gmail.com}
\author{Charles D. Brummitt}
\affiliation{University of California, Davis, California 95616}
\author{Raissa M. D'Souza}
\altaffiliation[Also at ]{The Santa Fe Institute, Santa Fe, NM 87501}
\affiliation{University of California, Davis, California 95616}

\date{October 20, 2013}
%
% 89.75.Hc Networks and genealogical trees
% 02.30.Yy Control theory 
% 05.65.+b Self-organized systems
% 45.70.Ht Avalanches [within 45.70.-n Granular systems]
\pacs{89.75.Hc, 02.30.Yy, 05.65.+b, 45.70.Ht}
%
%****************************************************************
\begin{abstract}
The Bak-Tang-Wiesenfeld (BTW) sandpile process is an archetypal, stylized model of complex systems with a critical point as an attractor of their dynamics. This phenomenon, called self-organized criticality (SOC), appears to occur ubiquitously in both nature and technology. Initially introduced on the 2D lattice, the BTW process has been studied on network structures with great analytical successes in the estimation of macroscopic quantities, such as the exponents of asymptotically power-law distributions. In this article, we take a microscopic perspective and study the inner workings of the process through both numerical and rigorous analysis. Our simulations reveal fundamental flaws in the assumptions of past phenomenological models, the same models that allowed accurate macroscopic predictions; we mathematically justify why universality may explain these past successes. Next, starting from scratch, we obtain microscopic understanding that enables mechanistic models; such models can, for example, distinguish a cascade's area from its size. In the special case of a $3$-regular network, we use self-consistency arguments to obtain a zero-parameters, mechanistic (bottom-up) approximation that reproduces nontrivial correlations observed in simulations and that allows the study of the BTW process on networks in regimes otherwise prohibitively costly to investigate. We then generalize some of these results to configuration model networks and explain how one could continue the generalization. The numerous tools and methods presented herein are known to enable studying the effects of controlling the BTW process and other self-organizing systems. More broadly, our use of multitype branching processes to capture information bouncing back-and-forth in a network could inspire analogous models of systems in which consequences spread in a bidirectional fashion.
\end{abstract}
\maketitle
%
%****************************************************************
\section{Introduction \label{section:introduction}}
Many complex systems affecting modern life, from infrastructure systems like power grids to the natural catastrophes that threaten them, appear to be poised near criticality. For instance, power law distributions seem to characterize the sizes of electrical blackouts~\cite{Dobson2007}, financial fluctuations~\cite{Gabaix2003}, neuronal avalanches~\cite{JuanicoJPhysA2007,Ribeiro2010,Haimovici2013}, earthquakes~\cite{Saichev2004}, landslides~\cite{Hergarten2003}, overspill in water reservoirs~\cite{Mamede2012}, forest fires~\cite{Malamud1998,SinhaRay2000} and solar flares~\cite{Lu1991,Paczuski2005}. Since its introduction in 1987~\cite{Bak1987,Bak1988}, the Bak-Tang-Wiesenfeld (BTW) sandpile has served as a useful paradigm for the self-organizing dynamics that may drive these systems toward critical points (also called self-organized criticality or SOC). The recent prevalence of cascading failures and overloads in networked infrastructures~\cite{Dobson2007,Barrett2012,Helbing2013} is one motivation for studying the sandpile model on random graphs~\cite{sandpile_wattsstrogatz_2d, sandpile_wattsstrogatz_1d, sandpile_ER_annealed, OFC_quenchednetwork, Goh2003, Lee2004,Goh2005, Brummitt2012, Lee2012}. Here we attain deeper understanding of the BTW sandpile model on networks, which provides important lessons and techniques for studies of self-organized critical processes in natural and engineered systems in general. 

We begin in Section~\ref{section:process} with a brief background on complex networks and on the BTW sandpile process on networks. We then provide a macroscopic understanding of how these systems self-organize, including a few past results of interest.

Though past work has significantly advanced our macroscopic understanding of the sandpile process on networks, we demonstrate in Sec.~\ref{section:motivation} some fundamental gaps between this macroscopic view and the microscopic reality. Sections~\ref{section:microscopic}--\ref{section:confmodel} seek to reconcile these macroscopic and microscopic perspectives. Before doing so, we emphasize in Sec.~\ref{subsection:motivation:selfconsistency} the need for our models of the sandpile process to be self-consistent, a feature that plays a central role in the rest of the paper.

In Section~\ref{section:microscopic}, we start from scratch and develop a microscopic, analytical understanding of the BTW sandpile process on networks. We prove rigorous results concerning the inner workings of a cascade, which are summarized in Theorems~\ref{theorem:constraintsroot}--\ref{theorem:AAtilde} and proved in the appendix.

Section~\ref{section:threereg} combines these rigorous results with self-consistency arguments to obtain a zero-parameters model for random $3$-regular networks. This model allows the study of the BTW sandpile process in a regime prohibitively costly to investigate in simulations. As reported in~\cite{Noel2013short}, this approach is, to our knowledge, the first analytical model that can separately calculate cascade size and cascade area.

Section~\ref{section:confmodel} generalizes to configuration model networks some of the results for random $3$-regular graphs in Sec.~\ref{section:threereg}, including the independent calculations of cascade size and cascade area. The same method as in Sec.~\ref{section:threereg} could be used to obtain a zero-parameters model. However, doing so appears more amenable to a case-by-case study (e.g., a mix of nodes of degree $3$ and $4$).

Section~\ref{section:conclusion} discusses the impact of our work in the general context of SOC on networks. We stress that self-consistency may be more important to the success of an analytical approximation---as it captures the self-organization mechanism---than direct consistency with the original system. We also argue that the success of our mechanistic, self-consistent modeling approach is an important proof of concept and that this method's applicability should extend beyond the BTW model.

Proofs and additional justifications are presented in appendix.
%
%****************************************************************
\section{The studied process \label{section:process}}
The BTW sandpile model considers a large collection of nodes that shed load to their neighbors whenever they reach their capacity. We drop a discrete unit of load (called a ``grain of sand'') randomly onto the system after each cascade finishes. A single grain can cause a large cascade (avalanche) of sand to move around the system as the system re-stabilizes. Such cascades typically occur in sizes distributed according to a power law because the system self-organizes to a critical point, as some have argued occurs to some extent in electrical grids~\cite{Dobson2007}, financial markets~\cite{Dupoyet2011}, neuronal avalanches~\cite{JuanicoJPhysA2007,Haimovici2013}, and some natural catastrophes~\cite{Saichev2004, Hergarten2003, SinhaRay2000, Malamud1998, Lu1991, Paczuski2005}. Inspired by the excess load and stress that can cascade among critical infrastructure systems---such as blackouts in power grids, patients in overwhelmed hospitals and excess travelers in transportation---here we study these cascades occurring on a network.

This section presents background material needed for the rest of the article. Section~\ref{subsection:process:networks} covers important concepts concerning complex networks, and Sec.~\ref{subsection:process:BTW} quickly describes the BTW process on networks. Finally, Sec.~\ref{subsection:process:macroscopic} explains how the BTW process self-organizes, and it presents the resulting observables for various network structures. 
%
%********************************
\subsection{Networks with quenched or annealed structure \label{subsection:process:networks}}
\emph{Networks} (graphs) consist of \emph{nodes} (vertices), representing the elements of a system, connected by \emph{links} (edges), representing interactions among those elements. Two nodes are \emph{neighbors} (adjacent) if they are joined by a link, and the \emph{degree} of a node is its number of neighbors. A network's \emph{degree distribution} is the sequence $\{p_k : k \geq 0\}$ such that $p_k$ is the fraction of nodes with degree $k$. The \emph{configuration model} samples random graphs from all graphs on $N$ nodes with a specified degree distribution. Our numerical implementations of the configuration model uses the following algorithm~\cite{Molloy1995,Newman2001}: (i) Create $N$ isolated nodes. (ii) Assign to each node a number of ``half-links'' sampled from the degree distribution $\{p_k : k \geq 0\}$. The total number of half-links should be even (otherwise, discard one node's degree and resample until the total number of half-links is even). (iii) Select two half-links uniformly at random and pair them to form a link. Repeat until no half-links remain. Although this process may create parallel edges and self-loops, they occur so rarely that they can be neglected.

This article studies networks with \emph{quenched structure}: once created, the structure of the network does not change. Thus, neighbors remain neighbors throughout the process. At the other extreme, the network could have \emph{annealed structure}, in which the structure of the network changes at a rate arbitrarily faster than the considered process. In this case, the network ``forgets'' the identities of past neighbors, and at each time step the network is an independent realization from the ensemble of networks. This distinction between quenched and annealed is important: quenched structure allows intricate correlations to appear among the internal states of nearby nodes, whereas such correlations among nodes' states are impossible in a network with annealed structure.
%
%********************************
\subsection{BTW sandpile model on networks \label{subsection:process:BTW}}
Originally introduced on the plane, the BTW sandpile process~\cite{Bak1987,Bak1988} can be generalized to networks in a few natural ways that differ only in specifics~\cite{sandpile_wattsstrogatz_2d, sandpile_wattsstrogatz_1d, sandpile_ER_annealed, OFC_quenchednetwork, Goh2003, Lee2004, Goh2005, Lee2012, Brummitt2012,Hoore2013}. Throughout this paper, we consider the following natural formulation~\cite{Goh2003, Lee2012, Brummitt2012}.

Each node holds grains of sand. We call a node \emph{$i$-sand} if it holds $i$ grains of sand. The \emph{capacity} of a node is the maximal amount of sand that it can hold. In this article, we set the capacity of every node to one less than its degree~\cite{sandpile_ER_annealed, Goh2003, Lee2004, Goh2005, Lee2012, Brummitt2012}. Hence, a $(k-1)$-sand node of degree $k$ is \emph{at capacity}, which means that it holds as much sand as it can withstand. Adding one or more grains to this node would bring it \emph{over capacity}. A node brought over capacity \emph{topples}, which means that it sheds its load by sending one grain to each of its neighbors.

The BTW sandpile process consists of a sequence of \emph{cascades} (avalanches), defined as follows. Drop a grain of sand on a node chosen uniformly at random, called the \emph{root} of the cascade. If this addition does not bring the root over capacity, then that cascade is finished. However, if the root is over capacity, then the root topples and sheds one grain to each of its neighbors. Any node that exceeds its capacity topples in the same way, until all nodes hold a number of grains less than or equal to its capacity. The \emph{size} of a cascade is the number of toppling events; the \emph{area} of a cascade is the number of nodes that topple. Subsequently, we begin a new cascade by dropping a grain on a root node chosen uniformly at random. Details of the simulation algorithm are provided in~\cite{SM_BTW_PRL_from_perspective_of_long}.

Some mechanism is required to avoid inundating the system with sand. In this article, we choose \emph{annealed dissipation}: whenever a grain of sand moves from one node to another, independently and with probability $\epsilon$ this grain disappears from the system. Another viable choice would be \emph{quenched dissipation}: a fraction $\epsilon$ of nodes are \emph{sinks}, and grains of sand sent to sinks disappear rather than pile up. (Note that the difference between these two alternatives is conceptually much less important than the one between networks with quenched structure and with annealed structure.)
%
%********************************
\subsection{Macroscopic understanding of the BTW sandpile process on networks \label{subsection:process:macroscopic}}
\begin{figure}
{\includegraphics{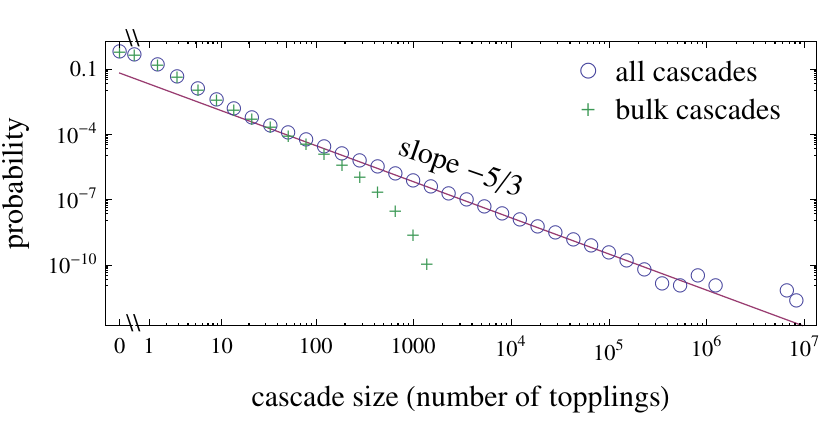}}
  \caption{(Color online)
  Cascade size distribution for a scale-free graph generated by the static model~\cite{Goh2001} with $N=10^7$ nodes, exponent $\gamma = 2.5$, average degree $6.3$, and dissipation rate $\epsilon = 10^{-3}$. We use only the largest component (containing $9503592$ nodes), and we collect statistics on the last $10^8$ cascades out of $1.2 \times 10^8$ cascades. We logarithmically bin the data for all cascades (empty circles) and for the ``bulk'' cascades in which no sand dissipates (``+'' symbols). Past work based on the $1/k$-assumption~\cite{Goh2003} successfully predicts that the cascade size distribution should follow a power law with exponent $-\gamma/(\gamma-1) = -5/3$ (line). \label{fig:powerlaw}}
\end{figure}
For a sufficiently large network with quenched structure ($N \to \infty$) and for a sufficiently small probability of annealed dissipation ($\epsilon \to 0$), the BTW sandpile process on networks self-organizes to a stationary state in which the cascade sizes follow a power law distribution. Figure~\ref{fig:powerlaw} demonstrates this well-known result using Monte Carlo simulations. Historically, only cascades in which no sand dissipates were analyzed~\cite{Goh2003, Lee2004, Goh2005}; the size distribution of these so-called ``bulk cascades'' are indicated by ``+'' symbols in Fig.~\ref{fig:powerlaw}. In this article, we focus on the distribution for \emph{all} cascades, illustrated by circles in Fig.~\ref{fig:powerlaw}. Note that the size distribution of bulk cascades can be approximated using that of all cascades by applying the exponentially decaying weight $\{(1-\epsilon)^{\langle k \rangle s}: s \geq 0\}$, where $s$ is the size of the cascade and $\langle k \rangle$ is the average node degree. This estimation of the chance that no grains dissipate uses the approximation that $\langle k \rangle s$ grains are shed during the cascade. Note that the cascades of size greater than $2\times 10^6$ are an artifact of the combined effect of the finite size of the system and the heavy-tailed degree distribution~\footnote{
%
% *** The footnote ***
In Fig.~\ref{fig:powerlaw}, the presence of a gap of width $5 \times 10^6$ in the cascade size distribution followed by a cluster of larger cascades appears to imply the existence of ``giant cascades''. However, this gap is a combined effect of the finite size of the system (here $N = 10^7$ nodes) and the heavy-tailed degree distribution (here scale-free with exponent $\gamma = 2.5$). For fixed $\epsilon$, the expected fractional size of ``giant cascades'' and the probability for such ``giant cascades'' both appear to converge to zero as $N \to \infty$, so these appear to not be ``real giants''. For a network with a light-tailed degree distribution (e.g., a random $3$-regular graph, in which every node has degree $3$), no such ``giant cascades'' are observed, even for finite $N$.%
}.

The power law behavior of the cascade size distribution indicates that the system is in a critical state: when a node topples, it directly causes on average $R_0 \approx 1$ other nodes to topple. An intuitive explanation reveals why the branching factor $R_0$ should approach unity. The crucial observation is that a cascade of size $s$ destroys on average $\kappa \epsilon s$ grains of sand, where $\kappa \approx \langle k \rangle$ is the expected number of grains shed by a node that topples. On one hand, if $R_0 < 1$, then the distribution of cascade size $s$ falls exponentially, so there exists an $\epsilon$ sufficiently small so that $\kappa \epsilon \langle s \rangle < 1$. Because one grain of sand is added at each cascade, and because each cascade destroys less than one grain on average, the amount of sand slowly builds up in the network, hence increasing $R_0$ toward unity. On the other hand, if $R_0 > 1$, then ``giant'' cascades with $s \propto N$ become possible. For fixed $\epsilon$, there exists an $N$ large enough so that $\kappa \epsilon \langle s \rangle > 1$, so more sand is destroyed than added, and $R_0$ thus decreases toward unity. Hence, the critical value $R_0 = 1$ is an attractor of the dynamics (SOC).

The value of the exponent $\tau$ of the power law cascade size distribution depends on the network structure. In the case of Fig.~\ref{fig:powerlaw}, a scale-free random graph with degree exponent $\gamma = 2.5$ was generated using the static model~\cite{Goh2001}, which results in an exponent $\tau = -5/3$ for the power law cascade size distribution. Using the assumption that at equilibrium a degree-$k$ node receiving a grain of sand has probability $1/k$ to topple---which we hereafter refer to as \emph{the $1/k$-assumption}---past work~\cite{Goh2003} predicts that for $2 < \gamma < 3$ one should obtain a cascade size distribution of exponent $\tau = \gamma/(\gamma-1)$. Empirical observations confirm this prediction. For $\gamma > 3$ and for other networks with light-tailed degree distributions, the ``mean-field'' value $\tau = -3/2$ is observed~\cite{sandpile_ER_annealed}. For instance, for a random $3$-regular graph, the degree distribution is light-tailed, so the slope $\tau = -3/2$; see Fig.~\ref{fig:size3reg}.

For networks with annealed structure, it has been shown that a node of degree $k$ selected uniformly at random is $i$-sand with probability $1/k$ for $0 \le i < k$~\cite{Christensen1993}, which implies that the $1/k$-assumption holds in this case. However, the proof hinges on the assumption of annealed structure, so it does not apply for networks with quenched structure. Past works report that the $1/k$-assumption approximately holds in simulations on networks with quenched structure~\cite{Goh2003, Lee2004, Goh2005, Brummitt2012, Lee2012}. Moreover, the $1/k$-assumption enables the aforementioned successful analytical prediction that $\tau = \gamma/(\gamma-1)$ for $2 < \gamma < 3$~\cite{Goh2003, Lee2004, Goh2005}. Notwithstanding this success, we show next in Sec.~\ref{section:motivation} that the full picture is more intricate.
%
%****************************************************************
\section{Motivation for further analysis \label{section:motivation}}
Past works on the BTW sandpile process on networks focus on the asymptotic behavior at the critical state. Within this paradigm, a great analytical success is the prediction of the power law exponent of the cascade size distribution for random networks with light tailed~\cite{sandpile_ER_annealed} and scale free~\cite{Goh2003, Lee2004, Goh2005} degree distributions. However, there are still fundamental gaps in our understanding of the BTW process on networks.

Section~\ref{subsection:motivation:oneoverkisbad} exposes some of these gaps, which motivate the in-depth analysis presented in the rest of this paper. Section~\ref{subsection:motivation:selfconsistency} justifies why our approach for this in-depth analysis places so much emphasis on self-consistency.
%
%********************************
\subsection{Caveats in the $1/k$-assumption: going beyond universality \label{subsection:motivation:oneoverkisbad}}
Although we confirm that the specific statement ``at equilibrium a degree-$k$ node receiving a grain of sand has probability $1/k$ to topple'' is approximately true in simulations, we identify three subtle points that are not appropriately reflected by this $1/k$-assumption.

First, denoting by $p_k$ the probability that a node selected uniformly at random has degree $k$, the assumption that a degree-$k$ node is $i$-sand with probability $1/k$ for $0 \le i < k$ predicts for the scale-free graph with exponent $\gamma = 2.5$ an average amount of sand per node of
\begin{align}
\sum_{k} \sum_{i=0}^{k-1} i \, \Prob \left(\text{degree $k$, $i$-sand} \right) = \sum_{k} p_k \sum_{i=0}^{k-1} i \frac{1}{k} & \approx 2.66 . \label{eq:averagesandif1overk}
\end{align}
However, the observed value in our simulations is $21\%$ higher (i.e., $3.22$ grains per node). For a random $3$-regular network, the difference is $50\%$ (namely, $1.50$ grains per nodes instead of the $\frac{0}{3} + \frac{1}{3} + \frac{2}{3} = 1$ grain per node predicted by the $1/k$-assumption). In all the networks that we studied, instead of being flatly distributed among the possible values $0 \le i < k$, the probability that a degree-$k$ node selected uniformly at random is $i$-sand is skewed toward larger values of $i$. In particular, we observed that the probability for this node to be at capacity [i.e., $(k-1)$-sand] is typically much larger than $1/k$.

This higher probability for a node to be at capacity is somehow counterbalanced by a second observation: nodes are unlikely to topple more than once during the same cascade, especially nodes that are many hops away from the root of this cascade. The reason is intuitive: a node that recently toppled had its amount of sand reset to zero at the moment of toppling, so it is unlikely to topple again during the same cascade. Hence, when a non-root node $u$ of degree $k$ topples, one of the $k$ grains shed by $u$ is unlikely to cause further topplings because that grain is sent to the node $v$ that caused $u$ to topple in the first place.

A third subtlety is that intricate correlations exist between the amounts of sand on neighboring nodes. For simplicity, we consider the example of a random $3$-regular graph. We define $\psi_i$ to be the probability that a uniformly random node is $i$-sand, and we define $\phi_{ij}$ to be the probability that a uniformly random neighbor of a uniformly random $i$-sand node is $j$-sand. Numerical simulations for a network of $N = 10^7$ nodes using a dissipation of $\epsilon = 10^{-3}$ provide the values
\begin{subequations}%
\begin{gather}%
\begin{pmatrix} \psi_0 & \psi_1 & \psi_2 \end{pmatrix}
\approx
\begin{pmatrix} 0.08397 & 0.33403 & 0.58199 \end{pmatrix} \label{eq:psi_simulation}, \\[1ex]
\begin{pmatrix}
 \phi_{00} & \phi_{01} & \phi_{02} \\
 \phi_{10} & \phi_{11} & \phi_{12} \\
 \phi_{20} & \phi_{21} & \phi_{22}
\end{pmatrix}
\approx
\begin{pmatrix} 0.00050 & 0.25148 & 0.74802 \\ 0.06321 & 0.31349 & 0.62330 \\ 0.10795 & 0.35773 & 0.53432 \end{pmatrix} . \label{eq:phi_simulation}
\end{gather}\label{eq:psiphi_simulation}%
\end{subequations}
If the amount of sand on neighboring nodes were uncorrelated, then we would expect $\phi_{ij}$ to be identically equal to $\psi_j$, independently of the value of $i$ on which we condition the probability. However, $\phi_{ij}$ deviates from $\psi_j$ by amounts ranging from $7\%$ to $33\%$, except for $\phi_{00}$, which deviates from $\psi_0$ by a factor of about $170$. The reason why $\phi_{00}$ is so low is intuitive: sand dissipates with small probability $\epsilon$ (e.g., $\epsilon = 10^{-3}$ in the above simulation), and $00$-sand links can appear only when a grain of sand dissipates as it is sent from a node to a $0$-sand neighbor.

There are also correlations of higher order than the pairwise correlations shown in Eq.~\eqref{eq:phi_simulation}, one of the simplest being ``$3$-star correlations'' (i.e., correlations between the amounts of sand on a node and on its $3$ neighbors). Consider, for instance, $\theta_{\tilde{k}}$, the probability for a $2$-sand node to have $\tilde{k}$ many $2$-sand neighbors out of its $3$ neighbors. If there were no $3$-star correlations, then one would expect the number of $2$-sand neighbors to be binomially distributed, i.e., $\theta_{\tilde{k}} = \binom{3}{\tilde{k}}(\phi_{22})^{\tilde{k}}(1-\phi_{22})^{3-\tilde{k}}$. Using $\phi_{22} \approx 0.53432$ obtained in Eq.~\eqref{eq:phi_simulation}, this hypothesis predicts
\begin{subequations}%
\begin{align}%
\!\!\bigl( \theta_0 \ \theta_1 \ \theta_2 \ \theta_3 \bigr) & \approx \bigl( 0.10099 \ \ 0.34761 \ \ 0.39885 \ \ 0.15255 \bigr) . \label{eq:motivation:oneoverkisbad:theta:binomial}
\end{align}%
However, simulations on a network of $N = 10^7$ nodes using a dissipation of $\epsilon = 10^{-3}$ provide the values
\begin{align}%
\!\!\bigl( \theta_0 \ \theta_1 \ \theta_2 \ \theta_3 \bigr) & \approx \bigl( 0.08359 \ \ 0.36941 \ \ 0.40744 \ \ 0.13956 \bigr) , \label{eq:motivation:oneoverkisbad:theta:measured}
\end{align}\label{eq:motivation:oneoverkisbad:theta}%
\end{subequations}%
which are significantly different from the predicted ones, revealing the presence of nontrivial $3$-star correlations. Nevertheless, ignoring the effects of $3$-star correlations by assuming independence beyond pairwise correlations is more acceptable than assuming no pairwise correlations (i.e., $\phi_{ij} \equiv \psi_j$). 

In view of these subtleties, it may seem surprising that the $1/k$-assumption performs so well at predicting the cascade size distribution's power law exponent $\tau = \gamma/(\gamma-1)$ for scale-free graphs with degree exponent $2 < \gamma < 3$. We conjecture that this success occurs because the branching process based on the $1/k$-assumption~\cite{Goh2003, Lee2004, Goh2005} belongs to the same universality class as the BTW process on a network. See Appendix~\ref{appendix:universality} for mathematical justifications of this conclusion.

Although the critical exponent $\tau$ appears to be unaffected by the aforementioned caveats of the $1/k$-assumption, there are measures affected by these caveats, such as the total amount of sand in the system, the cascade area distribution, and the details of the cascade size distribution (besides its tail behavior). Moreover, the $1/k$-assumption does not allow one to study the system outside of the critical regime ($\epsilon > 0$), nor how the system would behave if one were to control it~\cite{Noel2013short}. Sections~\ref{section:microscopic}--\ref{section:confmodel} take a completely different, ``bottom-up'' perspective that is not fundamentally subject to these limitations like the $1/k$-assumption is. Before this endeavor, Sec.~\ref{subsection:motivation:selfconsistency} covers one last point important for modeling self-organizing (and especially SOC) systems.
%
%********************************
\subsection{The importance of self-consistency in models of self-organized dynamics\label{subsection:motivation:selfconsistency}}

Given a dynamical system, a dynamical model of this system that uses some form of approximation may be understood as a different dynamical system. This apparently trivial statement has important implications in the presence of self-organization, especially for self-organized criticality. Although here we consider the special case where the ``original dynamical system'' is the BTW sandpile process, this general observation has wider relevance~\cite{Noel2013short}.

Consider the system given by the BTW sandpile process on a quenched graph $\mathcal{G}$ sampled from the ensemble of random $3$-regular graphs on $N$ nodes. Denote by $\mathcal{S}$ the rules for updating the system (i.e., dropping grains of sand, toppling nodes, and dissipating sand). Next, consider the system as a Markov chain over the set $\{0,1,2\}^N$ of all possible numbers of grains on each of the $N$ many nodes; we denote by $\widehat{\mathcal{S}}$ the vector of probabilities for the system to be in each possible state after infinitely many cascades have occurred. The triplet $(\mathcal{G},\mathcal{S},\widehat{\mathcal{S}}_n)$ specifies the full system after $n$ cascades have occurred. In the limit of many cascades, the state of the system approaches a stationary state (i.e., stationary distribution), denoted by $\widehat{\mathcal{S}}$.

As mentioned previously in Sec.~\ref{subsection:process:macroscopic}, the stationary state $\widehat{\mathcal{S}}$ approaches a critical point (i.e., the branching factor $R_0$ approaches $1$) in the limit of infinite system size $N \to \infty$ and vanishing dissipation $\epsilon \to 0$. Because of this criticality, the cascade size distribution for the system $(\mathcal{G}, \mathcal{S},\widehat{\mathcal{S}})$ has a power-law tail. Moreover, as mentioned in Sec.~\ref{subsection:motivation:oneoverkisbad}, the stationary state $\widehat{\mathcal{S}}$ contains correlations among the amounts of sand on nearby nodes of the network. 

Consider the Bethe lattice $\mathcal{B}$ with coordination number $3$ (i.e., an infinite $3$-regular graph without boundaries). In the limit $N \to \infty$, assuming that all cascades are finite, the finite system $(\mathcal{G},\mathcal{S},\widehat{\mathcal{S}})$ becomes indistinguishable from the infinite system $(\mathcal{B},\mathcal{S},\widehat{\mathcal{S}})$. In Sec.~\ref{section:threereg}, we define the model $(\mathcal{B},\mathcal{M},\widehat{\mathcal{M}})$ that approximates the infinite system $(\mathcal{B},\mathcal{S},\widehat{\mathcal{S}})$ by neglecting all higher order correlations beyond pairwise correlations. The equilibrium state $\widehat{\mathcal{M}}$ is fully specified by $\psi_i$ and $\phi_{ij}$ because, by definition, the model $\mathcal{M}$ cannot contain higher-order correlations. One could define another state $\widehat{\mathcal{M}}_{\widehat{\mathcal{S}}}$ by setting the model's values of $\psi_i$ and $\phi_{ij}$ to be equal to those measured in the stationary state $\widehat{\mathcal{S}}$. However, it is unlikely that this ``empirical state'' $\widehat{\mathcal{M}}_{\widehat{\mathcal{S}}}$ would be critical under the rules $\mathcal{M}$: among the numerous possible equilibrium states $\widehat{\mathcal{M}}$, rare are the states critical under the rules $\mathcal{M}$. Hence, the model $(\mathcal{B},\mathcal{M},\widehat{\mathcal{M}}_{\widehat{\mathcal{S}}})$ with empirically measured parameters would typically predict a cascade size distribution \emph{without} a power-law tail.

Since the rules $\mathcal{M}$ approximate the rules $\mathcal{S}$, it is very possible that they have a similar SOC mechanism. If this is the case (which we show to be true in Sec.~\ref{section:threereg}), then the equilibrium state $\widehat{\mathcal{M}}$ approaches a critical point as $\epsilon \to 0$ (note that the limit $N \to \infty$ is already accounted for in $\mathcal{B}$). Thus, like in the infinite-size, equilibrated system $(\mathcal{B},\mathcal{S},\widehat{\mathcal{S}})$, the model $(\mathcal{B},\mathcal{M},\widehat{\mathcal{M}})$ predicts a power-law cascade size distribution.

Hence, although the model with empirical parameters $\widehat{\mathcal{M}}_{\widehat{\mathcal{S}}}$ is ``closer'' to the equilibrated system $\widehat{\mathcal{S}}$ than the self-consistent, equilibrated model $\widehat{\mathcal{M}}$ is in terms of the $\psi_i$ and $\phi_{ij}$, the self-consistent model $(\mathcal{B},\mathcal{M},\widehat{\mathcal{M}})$ performs better at reproducing the behavior of the system $(\mathcal{B},\mathcal{S},\widehat{\mathcal{S}})$ than does the model with empirical parameters $(\mathcal{B},\mathcal{M},\widehat{\mathcal{M}}_{\widehat{\mathcal{S}}})$. By using a model that is self-consistent (i.e., in its own equilibrium state), we harness the power of the SOC mechanism encoded in the rules of that model. We expect that a similar statement (that letting the model reach equilibrium is probably more effective than forcing it to have the parameters observed in the real system) holds for network structures other than random $3$-regular graphs. Moreover, the self-organized system of interest need not be critical for the self-consistency of a model to be important, although the presence of a critical point induces a strong dependency in the model's parameters that amplifies the importance of self-consistency. Sections~\ref{subsection:threereg:changesYZ} and~\ref{subsection:threereg:bootstrap} present methods resulting in such a self-consistent model.
%
%****************************************************************
\section{Microscopic understanding of the BTW sandpile process on networks\label{section:microscopic}}

This section provides rigorous results concerning the inner workings of a sandpile cascade. We begin by characterizing in Sec.~\ref{subsection:microscopic:timesthatnodescantopple} the constraints affecting the number of times that nodes can topple in a cascade. Our main results are summarized in Theorem~\ref{theorem:constraintsroot} (for the root of any cascade) and Theorem~\ref{theorem:constraintsnonroot} (for non-root nodes in a tree-like cascade). In Corollary~\ref{corollary:rulespatterns}, we unpack the results of these theorems to make them useful for approximating cascades using branching processes in Secs.~\ref{section:threereg}--\ref{section:confmodel}. Using Theorems~\ref{theorem:constraintsroot}--\ref{theorem:constraintsnonroot}, we prove in Sec.~\ref{subsection:microscopic:areapercolation} sufficient conditions for a cascade to ``form a tree'', a case in which we obtain strong results. The main result, Theorem~\ref{theorem:AAtilde}, specifies sufficient conditions for using the full toolbox developed in Sec.~\ref{subsection:microscopic:timesthatnodescantopple}. An example use of Theorem~\ref{theorem:AAtilde} is Corollary~\ref{corollary:areabond}, which establishes a correspondence between cascade area and bond percolation on the subgraph induced by nodes at capacity. Other examples of using Theorem~\ref{theorem:AAtilde} are given in Secs.~\ref{section:threereg}--\ref{section:confmodel}. Proofs are deferred to the appendices.
%
%********************************
\subsection{Constraints applicable to tree-like cascades \label{subsection:microscopic:timesthatnodescantopple}}

In tree-like sandpile cascades on networks that are not necessarily tree-like, causality constrains the shape of cascades, the number of times that each node topples, and how many grains neighboring nodes exchange. The purpose of this section is to identify such constraints because they are essential to our later calculations of cascade area and size. Understanding Corollary~\ref{corollary:rulespatterns} and the associated Fig.~\ref{figure:patterns} suffices to understand the sections that follow, whereas Theorems~\ref{theorem:constraintsroot}--\ref{theorem:constraintsnonroot} are the main mathematical results that concisely summarize our characterization of the sandpile model on networks. Proofs are in Appendix~\ref{appendix:proofs:timesthatnodescantopple}.

Our strong statements (Theorems~\ref{theorem:constraintsroot}--\ref{theorem:constraintsnonroot}) require a notion of time and causality. To make time considerations well defined and independent of the numerical implementation of the sandpile process, we assume that each grain sent from one node to another takes a positive and possibly random amount of time to reach its target or to dissipate, and we assume that nodes topple as soon as they exceed their capacity. 
 
We are now ready to formulate strong constraints on the root (and on its surroundings) in any cascade on any network $\mathcal{G}$. Note that the phrase ``a node $v$ receives $n$ grains from a single neighbor'' has the intended meaning ``for at least one of $v$'s neighbors, $v$ received $n$ grains from that neighbor''. 

\begin{theorem}[Strong constraints on the root in any cascade]
In any cascade, the root topples $0$ times if and only if the root was not initially at capacity. Moreover, for any positive integer $n$, the root topples $n$ times by time $t \geq 0$ if and only if (a) the root was initially at capacity, and (b) by time $t$ the root received from each of its neighbors either $n$ or $n-1$ grains (except not $n$ grains from each of its neighbors). \label{theorem:constraintsroot}
\end{theorem}

Whereas Theorem~\ref{theorem:constraintsroot} provides strong constraints for the root in any cascade, Theorem~\ref{theorem:constraintsnonroot} provides strong constraints for non-root nodes in cascades that form a finite tree (which is the case, or nearly the case, for most cascades on tree-like networks, including configuration model networks). To precisely define this notion of a cascade forming a tree, for a cascade on the graph $\mathcal{G}$, we define the graph $\mathcal{G}^\dagger$ to have all the nodes of $\mathcal{G}$ that have sand sent toward them in the cascade and all the edges of $\mathcal{G}$ along which sand is sent in the cascade. That is, $\mathcal{G}^\dagger$ is the subgraph induced by the root, the nodes that topple, and the neighbors of the nodes that topple, from which we remove the links between pairs of nodes that both do not topple. We say that \emph{a cascade forms a finite tree $\mathcal{G}^\dagger$} if $\mathcal{G}^\dagger$ is a finite tree.

We introduce the following nomenclature for a cascade that forms a tree. When a node sends a grain to each of its neighbors, those neighbors are the node's \emph{children} in the cascade, and the toppled node is the \emph{parent} of these children. Only the first time at which a node receives sand during a cascade matters for this nomenclature. Hence, if a child later sends a grain to its parent, it does not acquire a new title (the child remains a child, and its parent remains its parent). We say that the root belongs to generation $0$, and the children of a node in generation $g$ belong to generation $g+1$. The chain of parents emanating from a node are its \emph{ancestors}, and the tree of children starting from one node are its \emph{descendants}. Hence, the root is the ancestor of every other node receiving sand during a cascade, and these nodes are all descendants of the root.

\begin{theorem}[Strong constraints on non-root nodes in cascades forming a finite tree]
For a cascade that forms a finite tree $\mathcal{G}^\dagger$, a non-root node $v$ topples $0$ times by time $t>0$ if and only if $v$ was not initially at capacity or $v$ receives $0$ grains from its parent by time $t$. Moreover, under the same conditions and for any positive integer $n$, a non-root node $v$ topples $n$ times by time $t>0$ if and only if all of the following conditions hold: (a) $v$ was initially at capacity; (b) $v$ received $n$ or $n+1$ grains from its parent by time $t$; and (c) $v$ received $n$ or $n-1$ grains from each of its children by time $t$ (except not $n$ grains from every child if $v$ received $n+1$ grains from its parent). \label{theorem:constraintsnonroot}
\end{theorem}

Together, Theorems~\ref{theorem:constraintsroot}--\ref{theorem:constraintsnonroot} provide strong constraints for any node in a cascade forming a finite tree, and this result holds at any time during such a cascade. Because the ultimate outcome of a cascade is of particular interest to our application, we specialize the mathematical apparatus of Theorems~\ref{theorem:constraintsroot}--\ref{theorem:constraintsnonroot} to the form in Corollary~\ref{corollary:rulespatterns}, which turns out to be useful when elaborating a branching process that predicts the outcome of a cascade. Before doing so, some definitions are required.

For a cascade that forms a finite tree $\mathcal{G}^\dagger$, we associate a \emph{pattern} to each node $v$ in $\mathcal{G}^\dagger$. A pattern inherits all properties of its associated node, such as whether the node is the root, whether it is at capacity, and whether it is a parent or child with respect to another node. 

Each non-root pattern has a \emph{signature} given by a pair of integers $(n, n')$, which characterizes grains exchanged between this pattern and its parent. Specifically, given a node $v$ with parent $u$, node $v$ is associated with a pattern of signature $(n,n')$ if and only if the parent $u$ sends $n$ grains toward $v$ and the parent $u$ receives $n'$ grains from $v$. Note that, due to dissipation, the child $v$ may receive fewer than $n$ grains from $u$, and $v$ may send more than $n'$ grains toward $u$. The intuition to keep in mind is that we count grains \emph{from the parent's perspective} (i.e., $n$ and $n'$ are the numbers of grains sent from and received by the parent $u$ with respect to this particular child $v$). Though each pattern has a single, well-defined signature, two different patterns may share the same signature. For simplicity, we say that a non-root node $v$ has signature $(n,n')$ if $v$ is associated to a pattern that has signature $(n,n')$.

We are now ready to provide the corollary of Theorems~\ref{theorem:constraintsroot}--\ref{theorem:constraintsnonroot} enumerating the rules (illustrated in Fig.~\ref{figure:patterns}) that enable our approximations of cascade size and area in Secs.~\ref{section:threereg}--\ref{section:confmodel}. 
\begin{figure}
\begin{center}
\includegraphics[width=\linewidth]{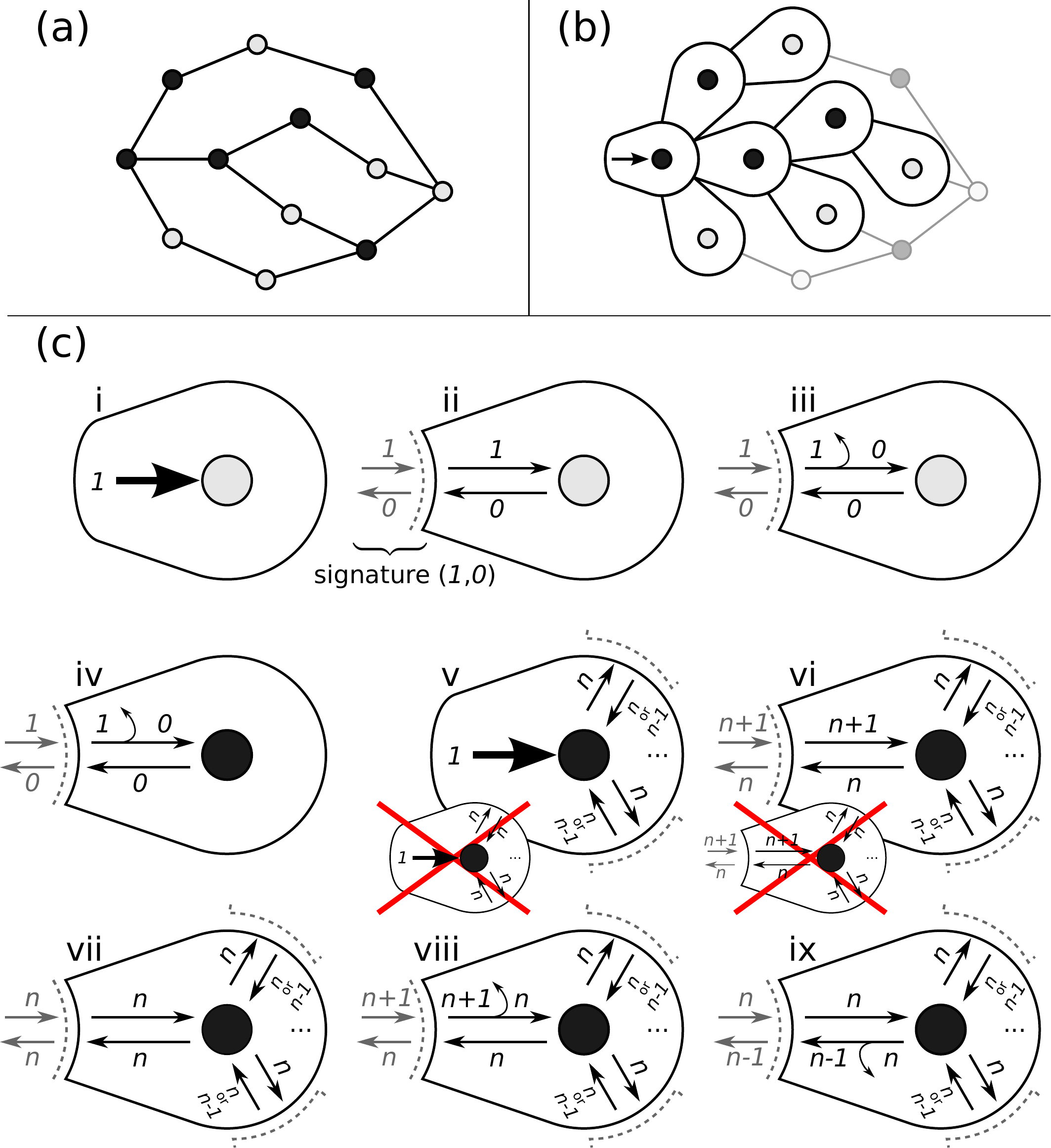} \\
\end{center}
\caption{(Color online) \label{figure:patterns}
Finite, tree-like cascades as an assemblage of patterns. (a) The graph before the cascade. Dark circles represent nodes initially at capacity, while light circles represent nodes initially not at capacity. (b) Starting at the root (indicated by an arrow), the cascade is assembled from patterns that encode how the grains of sand are shed during the process. Note that this assemblage may form a tree even though the original graph contained cycles. (c) The only possible patterns and rules for assembling them. The roman numerals correspond to the cases in Corollary~\ref{corollary:rulespatterns}. The convex shapes i and v represent root patterns, while the concave shapes represent non-root patterns. The large arrow represents the first grain dropped on the root; a curved arrow indicates dissipation; and the other arrows indicate exchanges of sand between neighboring nodes. The signature of a non-root pattern may be inferred by the arrows on its left. A large ``X'' marks a forbidden special case for patterns v and vi.}
\end{figure}
\begin{figure}
\begin{center}
\includegraphics[width=0.9\linewidth]{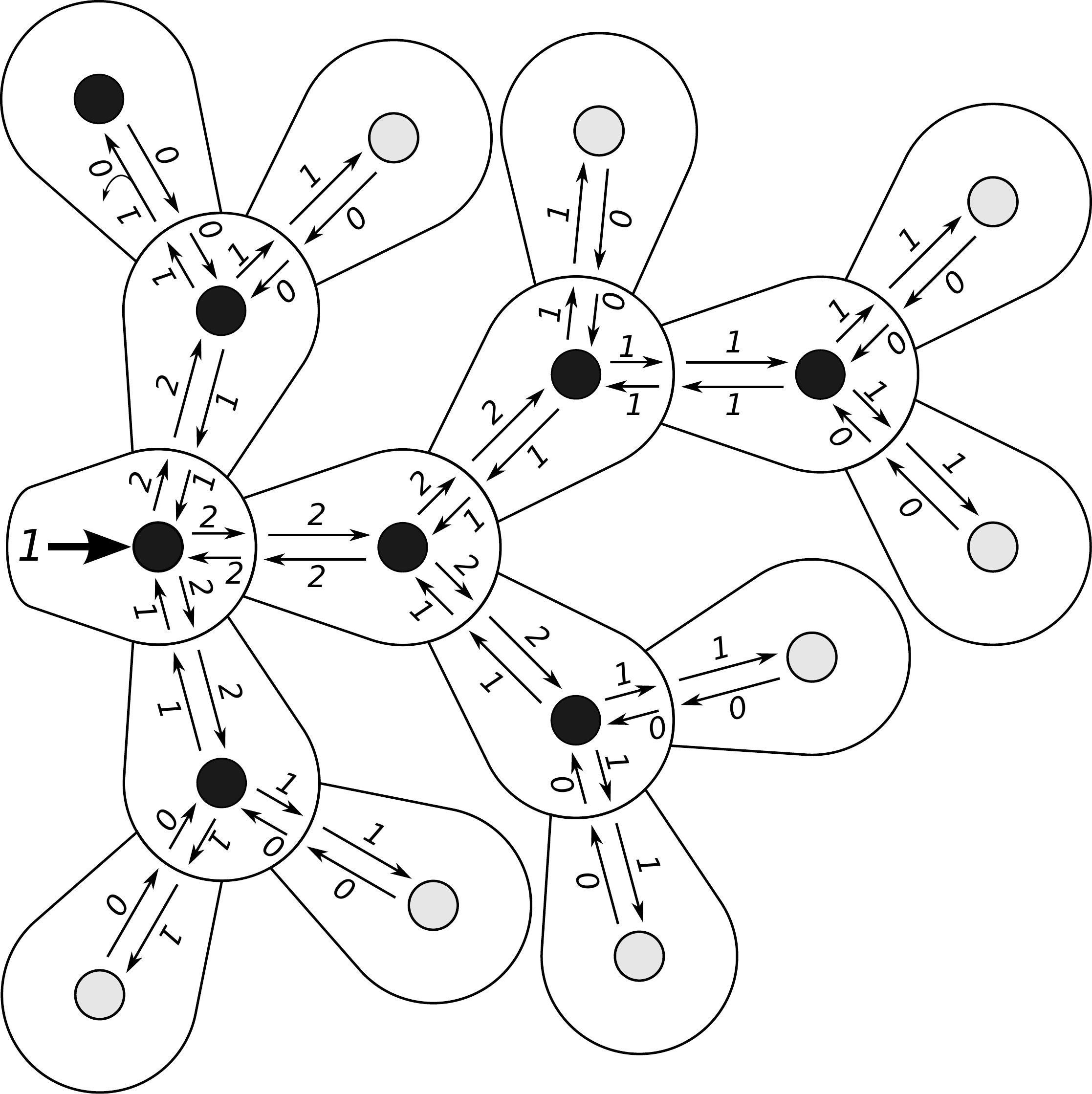}
\end{center}
\caption{
\label{figure:patterns_example}
Example of a cascade forming a finite tree $\mathcal{G}^\dagger$ from the ``pattern'' perspective. Refer to the caption of Fig.~\ref{figure:patterns} for notation. There are $16$ nodes in $\mathcal{G}^\dagger$, so $16$ patterns are required. All the $8$ nodes initially not at capacity do not topple [all pattern~\ref{pattern:nonrootnotatcapacityreceives}]. A single node at capacity does not topple (top left) because the grain sent towards it dissipates [pattern~\ref{pattern:nonrootatcapacitydissipates}]. The $7$ remaining nodes are all initially at capacity: $2$ nodes each topple twice [the root, center left, has pattern~\ref{pattern:rootatcapacity} and the node immediately to the right of the root has pattern~\ref{pattern:nonrootatcapacityreceivesnormal}] and the $5$ others topple once [$4$ have pattern~\ref{pattern:nonrootatcapacityreceivesspecial} and $1$ has pattern~\ref{pattern:nonrootatcapacityreceivesnormal}]. The resulting cascade area is $7$ and the cascade size is $9$.}
\end{figure}
\begin{corollary}[Constraints for patterns] \label{corollary:rulespatterns}
The following statements hold for a cascade that forms a finite tree $\mathcal{G}^\dagger$.

A node $v$ in $\mathcal{G}^\dagger$ topples $0$ times if and only if exactly one of the following holds:
\begin{enumerate}
\item $v$ is the root and is not initially at capacity. \label{pattern:rootnotatcapacity}
\item $v$ is non-root; $v$ is not initially at capacity; $v$ has signature $(1,0)$; and the grain sent by $v$'s parent toward $v$ reaches $v$. \label{pattern:nonrootnotatcapacityreceives}
\item $v$ is non-root; $v$ is not initially at capacity; $v$ has signature $(1,0)$; and the grain sent by $v$'s parent toward $v$ dissipates. \label{pattern:nonrootnotatcapacitydissipates}
\item $v$ is non-root; $v$ is initially at capacity; $v$ has signature $(1,0)$; and the grain sent by $v$'s parent toward $v$ dissipates. \label{pattern:nonrootatcapacitydissipates}
\end{enumerate}

For a positive integer $n$, a node $v$ in $\mathcal{G}^\dagger$ topples $n$ times if and only if exactly one of the following holds:
\begin{enumerate}[resume]
\item $v$ is the root, and each of $v$'s children has signature $(n,n)$ or $(n,n-1)$, except not all of its children may have signature $(n,n)$. \label{pattern:rootatcapacity}
\item $v$ is non-root; $v$ has signature $(n+1,n)$; and each of $v$'s children has signature $(n,n)$ or $(n,n-1)$, except not all of its children may have signature $(n,n)$. \label{pattern:nonrootatcapacityreceivesspecial}
\item $v$ is non-root; $v$ has signature $(n,n)$; and each of $v$'s children has signature $(n,n)$ or $(n,n-1)$.  \label{pattern:nonrootatcapacityreceivesnormal}
\item $v$ is non-root; $v$ has signature $(n+1,n)$; each of $v$'s children has signature $(n,n)$ or $(n,n-1)$; and the last grain sent by $v$'s parent toward $v$ dissipates. \label{pattern:nonrootatcapacitydissipatesin}
\item $v$ is non-root; $v$ has signature $(n,n-1)$; each of $v$'s children has signature $(n,n)$ or $(n,n-1)$; and the last grain sent by $v$ toward its parent dissipates. \label{pattern:nonrootatcapacitydissipatesout}
\end{enumerate}
\end{corollary}
Figure~\ref{figure:patterns_example} gives an explicit example of a cascade represented by the patterns~\ref{pattern:rootnotatcapacity}--\ref{pattern:nonrootatcapacitydissipatesout} enumerated in Corollary~\ref{corollary:rulespatterns}. 
%
%********************************
\subsection{Sufficient condition for cascades to form a tree on locally tree-like networks: a tool for calculating cascade area and size \label{subsection:microscopic:areapercolation}}

Equipped with the results of Sec.~\ref{subsection:microscopic:timesthatnodescantopple} regarding the possible patterns that may appear in a cascade, we now establish a sufficient condition for the cascade to form a tree, a case in which we can access the full power of Theorems~\ref{theorem:constraintsroot}--\ref{theorem:constraintsnonroot}. Readers interested only in the application of these rigorous results may skip this section with the message that, for ``locally tree-like'' networks, a finite cascade forms a finite tree, so only the nodes that are initially at capacity can topple during that cascade (Theorem~\ref{theorem:AAtilde}). A consequence of this result is that, under some reasonable conditions, obtaining the probability distribution for cascade area amounts to a bond percolation problem with bond occupation probability $1-\epsilon$ on the subgraph of the nodes that were at capacity before the cascade (Corollary~\ref{corollary:areabond}). Proofs are provided in Appendix~\ref{appendix:proofs:areapercolation}.

Now we make these messages precise. First, we make a definition for a node's surroundings to be ``locally tree-like''. Our definition is strict: we require all sufficiently small neighborhoods of a node to be a tree. Specifically, fix a positive integer $M$ and a node $v$ in a finite graph $\mathcal{G}$. We say that the triplet $(\mathcal{G}, M, v)$ is \emph{good} if, for every subgraph $U$ of $\mathcal{G}$ that contains node $v$ and at most $M-1$ other nodes, $v$ belongs to a component of $U$ that is a tree (i.e., that contains no cycles).

Denote by $\widetilde{\mathcal{G}}$ the subgraph of $\mathcal{G}$ induced by the nodes at capacity. Consider a cascade occurring on $\widetilde{\mathcal{G}}$ as a cascade occurring on $\mathcal{G}$ (starting with the same root) but with the nodes not at capacity treated as ``sinks'' that dissipate all sand sent to them. Theorem~\ref{theorem:AAtilde} specifies conditions guaranteeing that these nodes not at capacity would not topple even if they were not treated as sinks.

\begin{theorem}[Sufficient condition for the cascade to form a tree and for ignoring nodes not at capacity] \label{theorem:AAtilde}
Let $\mathcal{G}'$ be the subgraph induced by the nodes that topple in a cascade begun from the root $v$ on a graph $\mathcal{G}$, and let $\mathcal{A}$ be the area of this cascade. Let $\widetilde{\mathcal{G}}'$ be the subgraph induced by the nodes that topple in a cascade begun from the same root $v$ on the graph $\widetilde{\mathcal{G}}$ induced by the nodes at capacity in $\mathcal{G}$, and let $\widetilde{\mathcal{A}}$ be the area of this cascade. (If $v$ is not in $\widetilde{\mathcal{G}}$, define $\widetilde{\mathcal{G}}'$ as empty and $\widetilde{\mathcal{A}} = 0$.) 

If $M$ is an integer such that $(\mathcal{G},M,v)$ is good, and if $\widetilde{\mathcal{A}} \in \{0, 1, \dots,  M-1\}$, then the following hold: the cascade forms a finite tree $\mathcal{G}^\dagger$, $\widetilde{\mathcal{G}}'$ is a tree; $\mathcal{G}' = \widetilde{\mathcal{G}}'$; and $\mathcal{A} = \widetilde{\mathcal{A}}$.
\end{theorem}

We demonstrate the usefulness of Theorem~\ref{theorem:AAtilde} by the example of Corollary~\ref{corollary:areabond}: under the conditions of Theorem~\ref{theorem:AAtilde}, bond percolation is the same as cascade area. Denote by $\widetilde{\mathcal{G}}^{(1-\epsilon)}$ the graph that results from bond percolation on $\widetilde{\mathcal{G}}$ with bond occupation probability $1-\epsilon$. That is, $\widetilde{\mathcal{G}}^{(1-\epsilon)}$ is the subgraph of $\widetilde{\mathcal{G}}$ in which every link of $\widetilde{\mathcal{G}}$ has probability $1 - \epsilon$ to appear.

\begin{corollary}[Correspondence between cascade area and bond percolation on nodes at capacity] \label{corollary:areabond}
Suppose that a cascade begins at a node $v$ in a graph $\mathcal{G}$ whose subgraph of nodes at capacity is $\widetilde{\mathcal{G}}$. If $M$ is an integer such that $(\mathcal{G},M,v)$ is good, then for any $x \in \{0, 1, \dots,   M-1\}$, the chance that the cascade area $\mathcal{A}$ equals $x$ is identical to the chance that $v$ belongs to a component of size $x$ in $\widetilde{\mathcal{G}}^{(1-\epsilon)}$.
\end{corollary}

Corollary~\ref{corollary:areabond} enables a shortcut to estimating the cascade area distribution. Given a random root node $v$ in a graph $\mathcal{G}$ drawn from some ensemble of random graphs, if the neighborhood of $v$ is a tree (specifically, if $(\mathcal{G},M,v)$ is good), then Corollary~\ref{corollary:areabond} guarantees that the cascade area distribution is identical to the component size distribution in bond percolation on the subgraph of nodes at capacity, at least for cascade area smaller than $M$.

More generally, if $(\mathcal{G},M,v)$ is good, then by Theorem~\ref{theorem:AAtilde} a cascade on $\mathcal{G}$ starting from root $v$ and of area smaller than $M$ forms a finite tree $\mathcal{G}^\dagger$, so Theorems~\ref{theorem:constraintsroot}--\ref{theorem:constraintsnonroot} and Corollary~\ref{corollary:rulespatterns} hold. This conclusion allows the calculation of not only the cascade area but also of the inner workings of the cascade, including the cascade size.

We leave for future work the task of characterizing the probability that $(\mathcal{G},M,v)$ is good for various random graphs. Large configuration model random graphs (including random $3$-regular graphs) are locally tree-like, so for a random node $v$ it is likely that $(\mathcal{G},M,v)$ is good for reasonably large $M$. We explore such random graphs next.
%
%%****************************************************************
\section{Zero-parameters approximation of the BTW process on random $3$-regular graphs\label{section:threereg}}
Our goal in this section and the next is to calculate the probability distributions of cascade area and size for the BTW process on networks using the rigorous results of Sec.~\ref{section:microscopic} whenever possible. In this section, we focus on random $3$-regular networks and use self-consistency arguments to obtain a zero-parameters model that poises itself at a stationary state. 

Section~\ref{subsection:threereg:area} uses standard percolation methods (justified by Corollary~\ref{corollary:areabond}) to estimate the cascade area distribution using a single-type branching process. This step facilitates introducing the multitype branching process in Sec.~\ref{subsection:threereg:areasize} that harnesses the powerful results in Sec.~\ref{section:microscopic} to obtain both the cascade size and area distributions. Up to this point, the model has two parameters that need to be fixed (i.e., $\psi_2$ and $\phi_{22}$). In Sec.~\ref{subsection:threereg:changesYZ}, we further improve the model, so that it predicts the expected changes of its own parameters. Finally, Sec.~\ref{subsection:threereg:bootstrap} uses the assumption that these parameters reached an equilibrium, which results in our self-consistent, zero-parameters model.
%
%********************************
\subsection{``Standard'' percolation: cascade area\label{subsection:threereg:area}}
In order to facilitate the introduction of new techniques in the following sections, here we outline a rather standard percolation approach to calculate cascade area for the BTW sandpile process on a random $3$-regular network $\mathcal{G}$. The intuition behind the calculation is that if the neighborhood of the root is a tree and if the cascade area is small enough to fit inside this tree neighborhood, then Corollary~\ref{corollary:areabond} shows that cascade area is given by bond percolation on the subgraph $\widetilde{\mathcal{G}}$ of nodes at capacity.

Suppose, however, that $\widetilde{\mathcal{G}}$ is unknown to us, and that we only know three assumptions: (\emph{i}) $\mathcal{G}$ is a large random $3$-regular network; (\emph{ii}) a uniformly random node is at capacity with probability $\psi_2$; and (\emph{iii}) a node adjacent to a node at capacity is itself at capacity with probability $\phi_{22}$. We thus \emph{approximate} $\widetilde{\mathcal{G}}$ from this information, and we obtain an estimate of the cascade area distribution from the component size distribution of the graph $\widetilde{\mathcal{G}}^{(1-\epsilon)}$, the result of bond percolation on $\widetilde{\mathcal{G}}$ with bond occupancy $1-\epsilon$.

To this end, we use a standard single-type branching process approach~\cite{Newman2001,Newman2002} based on probability generating functions (PGFs). We define the cascade area PGF $\funH(x) \equiv \sum_{a = 0}^\infty \Prob(\text{area} = a) x^a$; the dummy variable $x$ is called the \emph{generator} for the cascade area. The function $\funH(x)$ is obtained by solving the system
\begin{subequations}%
\begin{align}%
  \funF(x) & = [1 - (1-\epsilon) \phi_{22}] + (1-\epsilon) \phi_{22} x \bigl[\funF(x)\bigr]^2 \label{eq:threereg:standard:F} \\
  \funH(x) & = (1 - \psi_2) + \psi_2 x \bigl[ \funF(x) \bigr]^3 . \label{eq:threereg:standard:H}
\end{align}\label{eq:threereg:standard}%
\end{subequations}%
The PGF $\funF(x)$ [Eq.~\eqref{eq:threereg:standard:F}] gives the contribution to the cascade area from a node sending a grain to a neighbor $v$ that has not yet toppled: the grain reaches node $v$ with probability $1-\epsilon$, and $v$ is at capacity with probability $\phi_{22}$. If both these events occur, then $v$ topples (factor $x$) and sends grains toward its $3$ neighbors, $2$ of which have not yet toppled (factor $[\funF(x)]^2$). In $\funH(x)$ [Eq.~\eqref{eq:threereg:standard:H}], the root is initially at capacity with probability $\psi_2$, in which case it topples (factor $x$) and sends a grain toward its $3$ neighbors (factor $[\funF(x)]^3$).
%
%********************************
\subsection{Inner workings: cascade size\label{subsection:threereg:areasize}}
We now improve Eq.~\eqref{eq:threereg:standard} to explicitly obtain the cascade size distribution in addition to the cascade area. This time, instead of Corollary~\ref{corollary:areabond}, we base our approach on the more fundamental Corollary~\ref{corollary:rulespatterns} and on a pattern representation of the cascade.

Keeping the generator $x$ for the cascade area, we introduce the generator $w$ for the cascade size such that $\funH(w,x) \equiv \sum_{a = 0}^\infty \sum_{s = 0}^\infty \Prob(\text{area} = a, \text{size} = s) x^a w^s$. Instead of $F(x)$ in Eq.~\eqref{eq:threereg:standard:F}, we need two families of functions, $\funA{}{n}(w,x)$ and $\funB{}{n}(w,x)$, which track the contributions of patterns of signature $(n,n-1)$ and $(n,n)$, respectively. Dropping explicit dependencies in $w$ and $x$ for simplicity, the system of equations becomes
\begin{subequations}%
\allowdisplaybreaks%
\begin{align}%
  \funA{}{1} & = \epsilon + \left(\frac{\epsilon}{1 - \epsilon}\right) \funB{}{1} + (1-\epsilon) (1-\phi_{22}) \label{eq:threereg:areasize:A1} \\
  \!\funA{}{n > 1} & = \epsilon \funB{}{n-1} + \left( \frac{\epsilon}{1 - \epsilon}\right) \funB{}{n} \label{eq:threereg:areasize:An} \\
  & + (1 - \epsilon)^{2n-1} \phi_{22} xw^{n-1} \left[ \bigl(\funA{}{n-1} + \funB{}{n-1} \bigr)^2 - \bigl( \funB{}{n-1} \bigr)^2 \right] \nonumber \\
  \funB{}{n} & = (1 - \epsilon)^{2n} \phi_{22} xw^n \bigl(\funA{}{n} + \funB{}{n} \bigr)^2 \label{eq:threereg:areasize:B} \\
  \funH & = (1-\psi_2) + \psi_2 x \sum_{\mathclap{n=1}}^\infty w^n \left[ \bigl(\funA{}{n} + \funB{}{n} \bigr)^3 - \bigl( \funB{}{n} \bigr)^3 \right]. \label{eq:threereg:areasize:H}
\end{align}\label{eq:threereg:areasize}%
\end{subequations}%
Next we explain how to obtain Eq.~\eqref{eq:threereg:areasize}. 

The factor $\funB{}{n}(w,x)$ [Eq.~\eqref{eq:threereg:areasize:B}] corresponds to the contribution of a node $v$ with pattern of signature $(n,n)$ [i.e., pattern~\ref{pattern:nonrootatcapacityreceivesnormal} in Corollary~\ref{corollary:areabond} and Fig.~\ref{figure:patterns}]. The factor $(1-\epsilon)^{2n}$ accounts for the probability that all the grains traveling from and to the parent of $v$ do not dissipate, while $\phi_{22}$ accounts for the probability that $v$ is at capacity at the beginning of the cascade (knowing that its parent was at capacity at the beginning of the cascade). The generators $xw^n$ count $v$'s contribution of $1$ to the cascade area and $n$ to the cascade size. Finally, the factor $(\funA{}{n} + \funB{}{n} )^2$ accounts for the contribution of the two children of $v$, each of which may have signature $(n,n-1)$ or $(n,n)$.

The three terms composing $\funA{}{1}(w,x)$ [Eq.~\eqref{eq:threereg:areasize:A1}] correspond to the contribution of nodes $v$ that have different patterns of signature $(1,0)$. The first term, $\epsilon$, accounts for the probability $\epsilon$ that the grain sent toward $v$ dissipates before reaching it [i.e., pattern~\ref{pattern:nonrootnotatcapacitydissipates} or~\ref{pattern:nonrootatcapacitydissipates} with $n = 1$], in which case there is no contribution to the cascade area nor size. The second term considers the possibility that $v$ topples, but the grain sent to its parent dissipates before reaching it [i.e., pattern~\ref{pattern:nonrootatcapacitydissipatesout}]: expanding that second term of Eq.~\eqref{eq:threereg:areasize:A1} gives
\begin{align}
  \left( \frac{\epsilon}{1 - \epsilon} \right) \funB{}{1} & = (1 - \epsilon) \epsilon \phi_{22} xw \bigl(\funA{}{1} + \funB{}{1} \bigr)^2 , \nonumber
\end{align}
where the factor $(1 - \epsilon) \epsilon$ accounts for the probability that $v$ receives the grain sent by its parent, but not the converse; and the rest of the expression is obtained in the same way as in $\funB{}{1}$. Finally, the third term of $\funA{}{1}$ corresponds to the pattern~\ref{pattern:nonrootnotatcapacityreceives}: $v$ receives the grain sent by its parent [factor $(1-\epsilon)$], but $v$ does not topple because it is initially not at capacity (factor $1 - \phi_{22}$).

Similarly, the three terms composing $\funA{}{n}(w,x)$ with $n > 1$ [Eq.~\eqref{eq:threereg:areasize:An}] correspond to the contribution of the different possible patterns of signature $(n,n-1)$. Here the first term corresponds to the case in which the last grain sent by the parent dissipates [i.e., pattern~\ref{pattern:nonrootatcapacitydissipatesin}], which amounts to $\funB{}{n-1}$ with a factor $\epsilon$ for the extra dissipation. Likewise, the second term corresponds to the case in which the last grain sent by $v$ to its parent dissipates [i.e., pattern~\ref{pattern:nonrootatcapacitydissipatesout}], which amounts to $\funB{}{n}$ with a factor $\epsilon/(1 - \epsilon)$ rectifying the probability for one extra sand dissipation and one fewer successful sand transfer. The last term of Eq.~\eqref{eq:threereg:areasize:An} corresponds to the pattern~\ref{pattern:nonrootatcapacityreceivesspecial}: none of the grains exchanged between $v$ and its parent dissipate [factor $(1-\epsilon)^{2n-1}$]; $v$ is initially at capacity (factor $\phi_{22}$); and $v$ topples $n-1$ times (factor $xw^{n-1}$). The remaining factor $[ (\funA{}{n-1} + \funB{}{n-1} )^2 - ( \funB{}{n-1} )^2 ]$ accounts for the contribution of the two children of this pattern, each of which may have signature $(n-1,n-2)$ or $(n-1,n-1)$, except not both may simultaneously have signature $(n-1,n-1)$.

Finally, $\funH(w,x)$ [Eq.~\eqref{eq:threereg:areasize:H}] tracks the contribution of root patterns. With probability $1-\psi_2$, the root is not initially at capacity, so it does not topple [i.e., pattern~\ref{pattern:rootnotatcapacity}]. Conversely, with probability $\psi_2$, the root is initially at capacity and topples $n \ge 1$ times (factor $x w^n$). The three children of this pattern~\ref{pattern:rootatcapacity} may have signature $(n,n-1)$ or $(n,n)$, but not all of them may have signature $(n,n)$ \{factor $[ (\funA{}{n} + \funB{}{n} )^3 - ( \funB{}{n} )^3 ]$\}.

One may check that setting $w = 1$ in Eq.~\eqref{eq:threereg:areasize} recovers Eq.~\eqref{eq:threereg:standard}: first verify that the ansatz $\funB{}{n}(1,x) = \funA{}{n+1}(1,x) + \funB{}{n+1}(1,x)$ satisfies Eq.~\eqref{eq:threereg:areasize}, then observe that $\funF(x) = \funA{}{1}(1,x) + \funB{}{1}(1,x)$ and that $\funH(x) = \funH(1,x)$ through a telescoping series.

Except for the structure of the network $\mathcal{G}$ and for the dissipation $\epsilon$ (with the case $\epsilon \to 0$ being of particular interest), the BTW sandpile process on network defined in Sec.~\ref{subsection:process:BTW} has no free parameter. However, both Eq.~\eqref{eq:threereg:standard} (based on a standard percolation approach) and the more refined Eq.~\eqref{eq:threereg:areasize} (accounting for the inner workings of a cascade) depend on two unknown quantities: $\psi_2$ and $\phi_{22}$. These two quantities are not ``real parameters'' of the BTW process; instead they reflect the state of the system once it has reached equilibrium.

As discussed in Sec.~\ref{subsection:motivation:selfconsistency}, one could measure $\psi_2$ and $\phi_{22}$ from the steady state of numerical simulations (i.e., $\widehat{\mathcal{S}}$), and then use these values in Eq.~\eqref{eq:threereg:standard} or in Eq.~\eqref{eq:threereg:areasize} to estimate the probability distributions for cascade area and/or size [which amounts to the model $(\mathcal{B},\mathcal{M},\widehat{\mathcal{M}}_{\widehat{\mathcal{S}}})$]. Proceeding in this way suffers the major disadvantage that results cannot be obtained outside the regime in which simulations were performed. Moreover, due to the reasons discussed in Sec.~\ref{subsection:motivation:selfconsistency}, the resulting model would be very inaccurate: using values from Eq.~\eqref{eq:psiphi_simulation}, both Eqs.~\eqref{eq:threereg:standard}--\eqref{eq:threereg:areasize} predict a branching factor of $R_0 = 2 (1-\epsilon) \phi_{22} \approx 1.07$, which is supercritical. Hence, feeding the model the $\psi_2$ and $\phi_{22}$ observed in simulations would predict the existence of ``giant cascades'' spanning a considerable fraction of the network in the limit $N \to \infty$, in contradiction with numerical simulations. (Note that the non-giant cascades would be affected by an exponential cutoff.) The next sections show how a self-consistency argument provides a zero-parameters model not subject to these disadvantages [which amounts to the model $(\mathcal{B},\mathcal{M},\widehat{\mathcal{M}})$].
%
%********************************
\subsection{Changes in the number of $i$-sand nodes and $ii$-sand links\label{subsection:threereg:changesYZ}}
Here we augment the multitype branching process in  Eq.~\eqref{eq:threereg:areasize} to obtain the effect of a cascade on the numbers of $i$-sand nodes and of $ii$-sand links in the network. This augmentation is a crucial step toward obtaining the zero-parameters model presented next: we will poise the model at an equilibrium by requiring that the quantities $\psi_i$ and $\phi_{ij}$ are on average unaffected by a cascade.

As before, $w$ and $x$ are generators for the cascade size and area, respectively. In addition, we define vectors of generators $\ve{y} = (y_0, y_1, y_2)$ and $\ve{z} = (z_0, z_1, z_2)$ such that $y_i$ and $z_i$ generate the changes in the numbers of $i$-sand nodes and of $ii$-sand links, respectively. Positive powers of $y_i$ or $z_i$ correspond to an increases in the respective counts, whereas negative powers correspond to decreases. For convenience, we define
\begin{align}%
  \funYy{i}{i\p} \equiv \frac{y_{i\p}}{y_i}, \ 
  \funZz{i}{j}{i\p}{j\p} \equiv \frac{1 + \delta_{i'j'}(z_{i'}-1)}{1 + \delta_{ij}(z_i-1)}, \ 
  \funLz{i}{i\p} = \sum_{j=0}^2 \phi_{ij} \funZ{i}{j}{i\p}{j}, \label{eq:threereg:YZL}
\end{align}%
where $\delta_{ij}$ is Kronecker's delta. A factor of $\funYy{i}{i\p}$ should be included whenever an $i$-sand node becomes $i'$-sand because $\funYy{i}{i\p}$ accounts for the change in the amount of sand on the node associated to this pattern (i.e., one fewer $i$-sand node and one more $i'$-sand node). Similarly, a factor of $\funZz{i}{j}{i\p}{j\p}$ should be included whenever an $ij$-sand link becomes $i'j'$-sand. Non-root patterns (whether they topple or not) are ``responsible'' for tracking the factor $\funZz{i}{j}{i\p}{j\p}$ due to the link joining them to their parent. We call a \emph{leaf} a node $v$ that does not receive sand but that has a neighbor $u$ that does receive sand (i.e., $v$ is a ``leaf'' of the branching process). If $u$ was initially $i$-sand and ends up being $i'$-sand after the cascade, then $v$ has probability $\phi_{ij}$ to initially be (and to remain) $j$-sand, so a factor $\funLz{i}{i\p}$ should be included for the link joining $u$ to $v$.

In order to account for the factor $\funZz{i}{j}{i\p}{j\p}$, a non-root pattern must know how many grains of sand end up being on its parent at the end of the cascade. Hence, we define the two families of functions $\funAxyz{i\p}{n}$ and $\funBxyz{i\p}{n}$ that track the contributions of nodes that have patterns of signature $(n,n-1)$ and $(n,n)$, respectively, and that have parents who end up being $i'$-sand after the cascade. The new PGF $\funHxyz$ is now obtained by solving a system of equations of structure very similar to Eq.~\ref{eq:threereg:areasize}
\begin{subequations}%
\allowdisplaybreaks%
\begin{align}%
  \begin{split}
  \funA{i\p}{1} & = \epsilon \funL{2}{i\p} + \left( \frac{\epsilon}{1 - \epsilon} \right) \funBip{1} \\
  & \quad + (1-\epsilon) \sum_{\mathclap{j=0}}^1 \phi_{2j} \funY{j}{j+1} \funZ{2}{j}{i'\!}{(j+1)} \bigl( \funL{j}{j+1} \bigr)^2
  \end{split}\label{eq:threereg:isandnodesiisandlinks:A1}\\%
  \begin{split}%
  \funA{i\p}{n>1} & = \epsilon \funBip{n-1} + \left( \frac{\epsilon}{1 - \epsilon} \right) \funBip{n} + (1 - \epsilon)^{2n-1} \phi_{22} xw^{n-1} \\
  & \quad \times \sum_{\mathclap{j'=1}}^2 \binom{2}{j\p-1} \funY{2}{j\p} \funZtop{j\p} \bigl(\funA{j\p}{n-1}\bigr)^{2-(j'-1)} \! \bigl( \funB{j\p}{n-1} \bigr)^{j'-1}
  \end{split}\label{eq:threereg:isandnodesiisandlinks:An}\\%
  \funB{i\p}{n} & = (1 - \epsilon)^{2n} \phi_{22} xw^n \! \sum_{\mathclap{j'=0}}^2 \binom{2}{j'} \funY{2}{j\p} \funZtop{j\p} \bigl( \funA{j\p}{n} \bigr)^{2-j\p} \bigl( \funB{j\p}{n} \bigr)^{j\p}\!\! \label{eq:threereg:isandnodesiisandlinks:B} \\
  \begin{split}%
  \funH & = \sum_{\mathclap{i=0}}^1 \psi_i \funY{i}{i+1} \bigl( \funL{i}{i+1} \bigr)^3 \\
  & \quad + \psi_2 x \sum_{\mathclap{n=1}}^\infty w^n \sum_{\mathclap{i'=0}}^2 \binom{3}{i\p} \funY{2}{i\p} \bigl( \funA{i\p}{n} \bigr)^{3-i\p} \bigl( \funB{i\p}{n} \bigr)^{i\p} . \label{eq:threereg:isandnodesiisandlinks:H}
  \end{split}%
\end{align}\label{eq:threereg:isandnodesiisandlinks}%
\end{subequations}%
Next we explain how to obtain Eq.~\eqref{eq:threereg:isandnodesiisandlinks}.

All the terms present in the equation for $\funB{}{n}$ [Eq.~\eqref{eq:threereg:areasize:B}] are still present in the equation for $\funB{i\p}{n}$ [Eq.~\eqref{eq:threereg:isandnodesiisandlinks:B}]. This time, we explicitly consider the signature of the two children of the node $v$ with pattern~\ref{pattern:nonrootatcapacityreceivesnormal}. Suppose $j' \in \{0,1,2\}$ of these children have signature $(n,n)$ and $2-j'$ have signature $(n,n-1)$. Summing the $2$ grains initially on the node $v$, the $n$ grains received from $v$'s parent, and the $nj' + (n-1)(2-j')$ received from $v$'s children, we obtain $3n+j'$ grains, enough for node $v$ to topple $n$ times and to end up $j'$-sand. A factor $\funY{2}{j\p}$ is required because $v$ was $2$-sand and becomes $j'$-sand. A factor $\funZtop{j\p}$ is required because the link from $v$ to its parent was $22$-sand and becomes $i'\!j'$-sand. The factor $\bigl( \funA{j\p}{n} \bigr)^{2-j\p} \bigl( \funB{j\p}{n} \bigr)^{j\p}$ accounting for $v$'s children passes on the information that $v$ ends up being $j'$-sand. The combinatorial factor $\binom{2}{j'}$ accounts for the number of ways to choose the signatures of the children.

Similarly, the terms present in the equation for $\funA{}{1}$ [Eq.~\eqref{eq:threereg:areasize:A1}] are also present in the equation for $\funA{i\p}{1}$ [Eq.~\eqref{eq:threereg:isandnodesiisandlinks:A1}]. The first term of Eq.~\eqref{eq:threereg:isandnodesiisandlinks:A1} corresponds to a leaf node [i.e., a node that does not receive any sand, pattern~\ref{pattern:nonrootnotatcapacitydissipates} or~\ref{pattern:nonrootatcapacitydissipates}], which requires a factor $\funL{2}{i\p}$. The second term is obtained similarly to $\funBip{1}$ in Eq.~\eqref{eq:threereg:isandnodesiisandlinks:B}. The third term corresponds to the pattern~\ref{pattern:nonrootnotatcapacityreceives}: for $j \in \{0,1\}$, the considered node $v$ is initially $j$-sand with probability $\phi_{2j}$, so $v$ receives a single grain (factor $\funY{j}{j+1}$) and does not topple. The link from $v$ to its parent mandates a factor $\funZ{2}{j}{i'\!}{(j+1)}$, and the two other neighbors of $v$ qualify as leaves, mandating a factor $\bigl( \funL{j}{j+1} \bigr)^2$.

Again, the terms present in the equation for $\!\funA{}{n > 1}$ [Eq.~\eqref{eq:threereg:areasize:An}] are also present in the equation for $\funA{i\p}{n>1}$ [Eq.~\eqref{eq:threereg:isandnodesiisandlinks:An}]. The first two terms are obtained similarly as in Eq.~\eqref{eq:threereg:isandnodesiisandlinks:B}. The last term corresponds to a node $v$ with pattern~\ref{pattern:nonrootatcapacityreceivesspecial}: $j\p-1 \in \{0,1\}$ of the $2$ children of $v$ have signature $(n-1,n-1)$, while the remaining $2-(j\p-1)$ children have signature $(n-1,n-2)$. Summing the $2$ grains initially on $v$, the $n$ grains received from $v$'s parent, and the $(n-1)(j\p-1) + (n-2)[2-(j\p-1)]$ grains received from $v$'s children, we obtain $3(n-1)+j'$ grains, enough for $v$ to topple $n-1$ times and end up $j'$-sand. Contributions similar to those in Eq.~\eqref{eq:threereg:isandnodesiisandlinks:B} are thus required.

Finally, the first term of the equation for $\funH$ [Eq.~\eqref{eq:threereg:isandnodesiisandlinks:H}] corresponds to the root not being at capacity [pattern~\ref{pattern:rootnotatcapacity}]. For $i \in \{0,1\}$, the root is $i$-sand with probability $\psi_i$ and thus becomes $(i+1)$-sand (factor $\funY{i}{i+1}$). Moreover, the $3$ neighbors of the root qualify as leaves [factor $\bigl (\funL{i}{i+1} \bigr)^3$]. The second term corresponds to the root being at capacity [pattern~\ref{pattern:rootatcapacity}]: $i' \in \{0,1,2\}$ of the root's children have signature $(n,n)$ and $3-i'$ have signature $(n,n-1)$. Summing the $2$ grains initially on the root, the $1$ grain dropped on the root to start the cascade, and the $ni' + (n-1)(3-i')$ grains received by the root from its children, we obtain $3n+i'$ grains, enough for the root to topple $n$ times and to end up $i'$-sand. The other factors follow the same logic as in Eq.~\eqref{eq:threereg:areasize:H}.

One may check that setting $y_i = z_i = 1$ for all $i \in \{0,1,2\}$ in Eq.~\eqref{eq:threereg:isandnodesiisandlinks} recovers Eq.~\eqref{eq:threereg:areasize}. Specifically, denoting $\ve{1} = (1 \ 1 \ 1)$, we have $\funY{i}{i\p}(\ve{1}) = \funZ{i}{j}{i\p}{j\p}(\ve{1}) = \funL{i}{i\p}(\ve{1}) = 1$, $\funA{i\p}{n}(w,x,\ve{1},\ve{1}) = \funA{}{n}(w,x)$ and $\funB{i\p}{n}(w,x,\ve{1},\ve{1}) = \funB{}{n}(w,x)$ for all the possible index values, which gives $\funH(w,x,\ve{1},\ve{1}) = \funH(w,x)$.
%
%********************************
\subsection{Bootstrapping $\psi_i$ and $\phi_{ij}$ \label{subsection:threereg:bootstrap}}
Now we obtain the self-consistent, zero-parameters model by enforcing equilibrium. The function $\funHxyz$ in Eq.~\eqref{eq:threereg:isandnodesiisandlinks:H} is a PGF of the multivariate probability distribution for the cascade size (generator $w$), the cascade area (generator $x$), the change in the number of $i$-sand nodes (generator $y_i$), and the change in the number of $ii$-sand links (generator $z_i$). By the standard techniques for PGFs, we obtain expectation values through differentiation with respect to the appropriate generator. In this case, we are interested in $h_{(i)}$, the average change in the number of $i$-sand nodes, and in $\eta_{(i)}$, the average change in the number of $ii$-sand links
\begin{align}%
\diffyH{i} & = \left. \frac{\partial \funH(1,1,\ve{y},\ve{1})}{\partial y_i} \right|_{\ve{y} = \ve{1}} &
\diffzH{i} & = \left. \frac{\partial \funH(1,1,\ve{1},\ve{z})}{\partial z_i} \right|_{\ve{z} = \ve{1}}.
\end{align}%
An explicit expression of these derivatives $\diffyH{i}$ and $\diffzH{i}$ is provided in~\cite{SM_BTW_PRL_from_perspective_of_long}.

By hypothesis, the system has reached a stationary state, so both these average changes should be zero. Hence, the relations
\begin{subequations}
\begin{align}%
\diffyH{i} = 0 \ \forall i \in \{ 0, 1, 2 \} \quad \text{and} \quad \diffzH{i} = 0 \ \forall i \in \{ 0, 1, 2 \} \label{eq:threereg:constraints:nonlin}
\end{align}%
provide $6$ constraints for the $12$ unknowns $\psi_i \, \forall i \in \{ 0, 1, 2 \}$ and $\phi_{ij} \, \forall i,j \in \{ 0, 1, 2 \}$. Only $5$ of these constraints are independent because each of the $\diffyH{i} = 0$ relations can be obtained from the two others. Because the $\psi_i$ are probabilities and the $\phi_{ij}$  are conditional probabilities, they must obey the $7$ additional, independent constraints
\begin{align}%
  \sum_{i = 0}^2 \psi_i = 1 , \ 
  \sum_{j = 0}^2 \phi_{ij} = 1 \ \forall i , \text{ and} \ 
  \psi_i \phi_{ij} = \psi_j \phi_{ji} \ \forall i,j.
\end{align}\label{eq:threereg:constraints}%
\end{subequations}%
Thus, we have a total of $12$ unknowns and $12$ independent constraints on them. Because the constraints~\eqref{eq:threereg:constraints:nonlin} are non-linear, \textit{a priori} there is no guarantee that a valid solution exists, nor that there is a single solution.

However, starting from educated guesses, numerical solution of the system~\eqref{eq:threereg:constraints} does provide values of $\psi_i$ and $\phi_{ij}$ that are consistent with those observed in the Monte Carlo simulations. For example, the solution of Eq.~\eqref{eq:threereg:constraints} for $\epsilon = 10^{-3}$ gives
\begin{subequations}
\begin{align}%
\begin{pmatrix} \psi_0 & \psi_1 & \psi_2 \end{pmatrix} 
& \approx \begin{pmatrix} 0.09574 & 0.35439 & 0.54987 \end{pmatrix}, \\
\begin{pmatrix}
 \phi_{00} & \phi_{01} & \phi_{02} \\
 \phi_{10} & \phi_{11} & \phi_{12} \\
 \phi_{20} & \phi_{21} & \phi_{22} \\
\end{pmatrix}
& \approx 
\begin{pmatrix} 0.00050 & 0.27802 & 0.72148 \\ 0.07511 & 0.34286 & 0.58203 \\ 0.12562 & 0.37512 & 0.49926 \end{pmatrix}.
\end{align}\label{eq:psiphi_analytical001}%
\end{subequations}%
Comparison with the results from simulations in Eq.~\eqref{eq:psiphi_simulation} reveals that, while the correspondence is not exact, a large fraction of the previously unexplained correlations is now accounted for. Part of the deviations is explained by the fact that Eq.~\eqref{eq:psiphi_simulation} is obtained from simulations with $N = 10^7$ nodes, while Eq.~\eqref{eq:psiphi_analytical001} assumes $N \to \infty$. As discussed in Sec.~\ref{subsection:motivation:selfconsistency}, other deviations are likely due to Eq.~\eqref{eq:psiphi_analytical001} only considering pairwise correlations, while correlations between the amount of sand on nodes arbitrarily far from one another could appear in the network simulations, especially for small $\epsilon$. Section~\ref{subsection:confmodel:areasize} provides guidelines on how future work could consider such higher-order correlations by obtaining a more refined version of Eq.~\eqref{eq:threereg:isandnodesiisandlinks}.

\begin{figure}
  \includegraphics{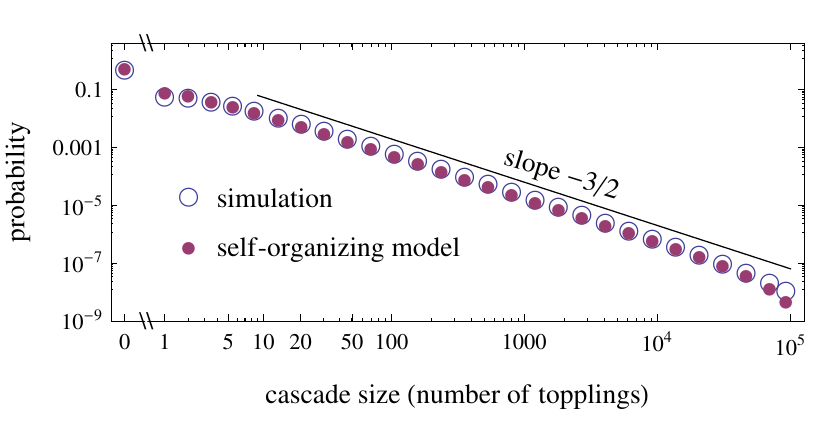}
  \caption{(Color online) Cascade size distribution for a $3$-regular network with $N = 10^7$ nodes and dissipation rate $\epsilon = 10^{-3}$. The data for both Monte Carlo simulations (open circles) and theory (filled circles) have been logarithmically binned. The theory is obtained for $\psi_i$ and $\phi_{ij}$ shown in Eq.~\ref{eq:psiphi_analytical001}, i.e., those satisfying the constraints [Eq.~\eqref{eq:threereg:constraints}] of the self-organizing model.\label{fig:size3reg}}
\end{figure}

Using $\psi_2 \approx 0.54987$ and $\phi_{22} \approx 0.49926$ [i.e., the values from Eq.~\eqref{eq:psiphi_analytical001}] in Eq.~\eqref{eq:threereg:areasize} provides estimates for the cascade area and size. Figure~\ref{fig:size3reg} confirms the accuracy of the resulting cascade size distribution by comparing it to numerical simulations. Because cascade size and area are often close to one another, we verify in Fig.~\ref{fig:sizeminusarea3reg} that Eq.~\eqref{eq:threereg:areasize} predicts the right distribution for the difference of a cascade's size and area. To the best of our knowledge, this approach is the first to allow the independent study of cascade size and area.

\begin{figure}
  \includegraphics{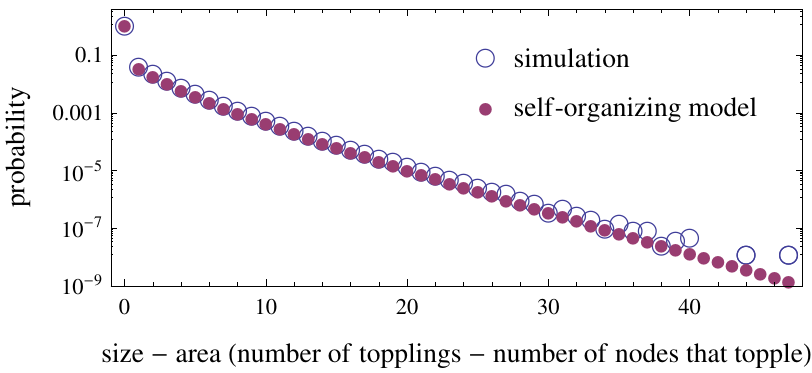}
  \caption{(Color online) Distribution of the difference between cascade size and cascade area for a $3$-regular network in the same conditions as in Fig.~\ref{fig:size3reg}. We consider the difference between the two quantities to facilitate the comparison. Note that the size of a cascade is always greater or equal to its area.\label{fig:sizeminusarea3reg}}
\end{figure}

In the range $0 < \epsilon < 0.01$, the $\psi_i$ and $\phi_{ij}$ show close-to-linear behavior (e.g., $\phi_{00} \approx \epsilon / 2$), and the total variation of each of these probabilities is less than $0.01$. Linear extrapolation provides the limiting behavior $\epsilon \to 0$
\begin{subequations}%
\begin{align}%
\lim_{\epsilon \to 0}
\begin{pmatrix} \psi_0 & \psi_1 & \psi_2 \end{pmatrix} 
& \approx \begin{pmatrix} 0.09519 & 0.35406 & 0.55075 \end{pmatrix} \\
\lim_{\epsilon \to 0}
\begin{pmatrix}
 \phi_{00} & \phi_{01} & \phi_{02} \\
 \phi_{10} & \phi_{11} & \phi_{12} \\
 \phi_{20} & \phi_{21} & \phi_{22} \\
\end{pmatrix}
& \approx 
\begin{pmatrix}
  0.00000 & 0.27676 & 0.72324 \\
  0.07441 & 0.34226 & 0.58333 \\
  0.12500 & 0.37500 & 0.50000
\end{pmatrix} \label{eq:analytical:philimit}
\end{align}\label{eq:analytical:psiphilimit}%
\end{subequations}
These results support our intuition that there should be no $00$-links in the limit $\epsilon \rightarrow 0$ and that the system tunes itself to the onset of a giant component (which happens at $\phi_{22} = 1/2$ when $\epsilon \rightarrow 0$). Other quantities appear to approach ratios of small integers: $\phi_{20} \approx 1/8$, $\phi_{21} \approx 3/8$, and $\phi_{12} \approx 7/12$. Note that these rational numbers are indicative only: excluding $\phi_{00}$ and $\phi_{22}$, we have no reason to offer as to why the $\phi_{ij}$ should be close to the ratio of small integers (and this proximity to simple ratios may well be coincidental).
%
%%****************************************************************
\section{Generalization to configuration model networks \label{section:confmodel}}
Here we show how the methods presented in Sec.~\ref{section:threereg} for $3$-regular networks generalize to configuration model networks. Hence, Sec.~\ref{subsection:confmodel:area} (resp.\ Sec.~\ref{subsection:confmodel:areasize}) repeats the treatment of Sec.~\ref{subsection:threereg:area} (resp.\ Sec.~\ref{subsection:threereg:areasize}) to predict the cascade area distribution (resp.\ cascade size and area distributions) on configuration model networks. However, we do not perform the two final steps (Secs.~\ref{subsection:threereg:changesYZ}--\ref{subsection:threereg:bootstrap}) that would provide a zero-parameters model: doing so reveals to be tedious in the case of the fully general configuration model. Nonetheless, future work could apply the same approach on a case-by-case basis to other network structures of interest.
%
%********************************
\subsection{``Standard'' percolation: cascade area \label{subsection:confmodel:area}}

We consider the same problem as in Sec.~\ref{subsection:threereg:area}---that is, to estimate the cascade area distribution using standard percolation techniques---but this time we consider the general case of a configuration model random graph $\mathcal{G}$ instead of a random $3$-regular graph. Again, we base our model on Corollary~\ref{corollary:areabond}: if the neighborhood of the root is a tree and if the cascade area is small enough to fit inside this tree neighborhood, then the cascade area is given by bond percolation on the subgraph $\widetilde{\mathcal{G}}$ of nodes at capacity. This preliminary step facilitates the presentation of Sec.~\ref{subsection:confmodel:areasize}, which considers the inner workings of the cascade to estimate both the cascade area and size distribution.

As described in Sec.~\ref{subsection:process:networks}, the configuration model random graph $\mathcal{G}$ is drawn from an ensemble specified by the degree distribution $\{p_k : k \geq 0\}$ and by $N$, the number of nodes in $\mathcal{G}$. We now make the assumption (approximation) that $\widetilde{\mathcal{G}}$, which is the subgraph of $\mathcal{G}$ induced by the nodes at capacity, is also a configuration model random graph. Specifically, we assume that $\widetilde{\mathcal{G}}$ is a configuration model random graph drawn from an ensemble specified by the degree distribution $\{\widetilde{p}_{\tilde{k}} : \tilde{k} \geq 0\}$ and by $\sigma N$, the number of nodes in $\widetilde{\mathcal{G}}$. Stated differently, $\sigma$ is the fraction of nodes in $\mathcal{G}$ that are at capacity at the beginning of the cascade, and $\widetilde{p}_{\tilde{k}}$ is the probability that a random node at capacity in $\mathcal{G}$ has $\tilde{k}$ neighbors at capacity in $\mathcal{G}$.

Following Newman~\cite{Newman2001,Newman2002}, we define $\widetilde{q}_{\tilde{k}}$ as the probability that a node in $\widetilde{\mathcal{G}}$ reached by following a random link of $\widetilde{\mathcal{G}}$ has degree $\tilde{k} + 1$ (i.e., has $\tilde{k}$ links in $\widetilde{\mathcal{G}}$ other than the one along which we reached it). Using the standard calculation for configuration model random graphs, we obtain $\widetilde{q}_{\tilde{k}}$ from $\widetilde{p}_{\tilde{k}}$ as $\widetilde{q}_{\tilde{k}} = ({\tilde{k}} + 1) \widetilde{p}_{\tilde{k} + 1}/\sum_{\tilde{k}'} \tilde{k}' \widetilde{p}_{\tilde{k}'}$. For convenience, we define the PGFs
\begin{align}%
  \funPt(\xi) & := \sum_{\tilde{k} = 0}^\infty \widetilde{p}_{\tilde{k}} \xi^{\tilde{k}}, &
  \funQt(\xi) & := \sum_{\tilde{k} = 0}^\infty \widetilde{q}_{\tilde{k}} \xi^{\tilde{k}} = \frac{ \frac{\diff\funPt(\xi)}{\diff \xi} }{\left.\frac{\diff\funPt(\xi)}{\diff\xi} \right|_{\xi = 1}}, \label{eq:confmodel:PQ}
\end{align}%
which encode the same information as the probability distributions $\widetilde{p}_{\tilde{k}}$ and $\widetilde{q}_{\tilde{k}}$ of degree and excess degree in $\widetilde{\mathcal{G}}$. We use these functions in Eq.~\eqref{eq:confmodel:PQ} to build a branching process providing the component sizes in $\widetilde{\mathcal{G}}^{(1-\epsilon)}$. By adapting~\cite{Newman2002}, we obtain the PGFs
\begin{subequations}%
\begin{align}%
  \funG(x) & = \epsilon + (1-\epsilon) x \funQt\biglb( \funG(x) \bigrb) \label{eq:confmodel:standard:F} \\
  \funH(x) & = (1 - \sigma) + \sigma x \funPt\biglb( \funG(x) \bigrb) . \label{eq:confmodel:standard:H}
\end{align}\label{eq:confmodel:standard}%
\end{subequations}%
The PGF $\funG(x)$ [Eq.~\eqref{eq:confmodel:standard:F}] generates the probability distribution for the size of the component of $\widetilde{\mathcal{G}}^{(1-\epsilon)}$ reached by following a random link of $\widetilde{\mathcal{G}}$. Note that this link in $\widetilde{\mathcal{G}}$ is also in $\widetilde{\mathcal{G}}^{(1-\epsilon)}$ with probability $1 - \epsilon$, so this component size equals zero with probability $\epsilon$. The full component size distribution of $\widetilde{\mathcal{G}}^{(1-\epsilon)}$ is generated by $x \funPt\biglb( \funG(x) \bigrb)$. A cascade of area zero occurs with probability $1 - \sigma$ (the fraction of nodes not at capacity). Hence, $\funH(x)$ [Eq.~\eqref{eq:confmodel:standard:H}] is the PGF for our approximation of the distribution of the area of cascades in $\mathcal{G}$.

One may verify that setting $\sigma = \psi_2$, $\funPt(\xi) = [(1-\phi_{22})+\phi_{22}\xi]^3$ and $\funQt(\xi) = [(1-\phi_{22})+\phi_{22}\xi]^2$ in Eq.~\eqref{eq:confmodel:standard} recovers the $3$-regular special case Eq.~\eqref{eq:threereg:standard}, with the correspondence $\funF(x) = (1-\phi_{22}) + \phi_{22}\funG(x)$. 
%
%********************************
\subsection{Inner workings: cascade size \label{subsection:confmodel:areasize}}

This section is to Sec.~\ref{subsection:confmodel:area} as Sec.~\ref{subsection:threereg:areasize} is to Sec.~\ref{subsection:threereg:area}: we improve Eq.~\ref{eq:confmodel:standard} by accounting for the inner workings of a cascade to explicitly obtain both the cascade area and size distribution (analogously to our improvement to Eq.~\eqref{eq:threereg:standard} in Sec.~\ref{subsection:threereg:areasize} for random $3$-regular graphs). In addition to $\funPt(\xi)$ and $\funQt(\xi)$ [Eq.~\eqref{eq:confmodel:PQ}], we define the related PGFs
\begin{equation}%
\funPts(\xi) = \sum_{\tilde{k}=0}^\infty \widetilde{p}_{\tilde{k}}^{\,*} \xi^{\tilde{k}} \quad \text{and} \quad
\funQts(\xi) = \sum_{\tilde{k}=0}^\infty \widetilde{q}_{\tilde{k}}^{\,*} \xi^{\tilde{k}}.
\end{equation}
Here $\widetilde{p}_{\tilde{k}}^{\,*}$ is the probability that a node at capacity has $\tilde{k}$ neighbors in $\mathcal{G}$ and that all $\tilde{k}$ neighbors are at capacity. Similarly, $\widetilde{q}_{\tilde{k}}^{\,*}$ is the probability for a node at capacity to have $\tilde{k} + 1$ neighbors in $\mathcal{G}$ that are all at capacity, given that this node was reached by following a link from a node that was at capacity. Note that, unlike $\funPt(1)$ and $\funQt(1)$, neither $\funPts(1)$ nor $\funQts(1)$ will typically equal one. The purpose of $\funPts(\xi)$ and $\funQts(\xi)$ is to properly consider particular constraints imposed by Corollary~\ref{corollary:rulespatterns}, such as the fact that the children of patterns~\ref{pattern:rootatcapacity} and~\ref{pattern:nonrootatcapacityreceivesspecial} may not all simultaneously have signature $(n,n)$.

We again use the generator $x$ to track the cascade area and the generator $w$ to track the cascade size. We define the two families of PGFs $\funC{n}(w,x)$ and $\funD{n}(w,x)$ that are related to the families $\funA{}{n}(w,x)$ and $\funB{}{n}(w,x)$, respectively, except that we are now only considering nodes at capacity (i.e., nodes in $\widetilde{\mathcal{G}}$). Specifically, $\funC{n}(w,x)$ [resp.\ $\funD{n}(w,x)$] accounts for the contribution of a non-root pattern at capacity of signature $(n,n-1)$ [resp.\ $(n,n-1)$], i.e., patterns~\ref{pattern:nonrootatcapacitydissipates}, \ref{pattern:nonrootatcapacityreceivesspecial}, \ref{pattern:nonrootatcapacitydissipatesin}, and \ref{pattern:nonrootatcapacitydissipatesout} [resp.\ pattern~\ref{pattern:nonrootatcapacityreceivesnormal}]. The relations defining $\funH(w,x)$ are
\begin{subequations}%
\allowdisplaybreaks%
\begin{align}%
  \funC{1} & = \epsilon \left( 1 + \frac{\funD{1}}{1-\epsilon} \right) \label{eq:confmodel:areasize:C1} \\
  \funC{2} & = \epsilon \left( \!\funD{1} + \frac{\funD{2}}{1-\epsilon} \right) + (1-\epsilon)^3 xw \Bigl[ \funQt\bigl( \funC{1} + \! \funD{1} \bigr) - \funQts\bigl( \funD{1} \bigr) \Bigr] \label{eq:confmodel:areasize:C2} \\
  \!\!\!\funC{n>2} & = \epsilon \left( \funD{n-1} + \frac{\funD{n}}{1-\epsilon} \right) \label{eq:confmodel:areasize:Cn} \\
  & \ + (1-\epsilon)^{2n-1} xw^{n-1} \Bigl[ \funQts\bigl( \funC{n-1} + \funD{n-1} \bigr) - \funQts\bigl( \funD{n-1} \bigr) \Bigr] \nonumber \\
  \funD{1} & = (1-\epsilon)^2 x w \funQt\bigl( \funC{1} + \funD{1} \bigr) \label{eq:confmodel:areasize:D1} \\
  \!\!\!\funD{n>1} & = (1-\epsilon)^{2n} x w^n \funQts\bigl( \funC{n} + \funD{n} \bigr) \label{eq:confmodel:areasize:Dn} \\
  \begin{split}
  \funH & = (1-\sigma) + \sigma x \biggl\{ w \Bigl[ \funPt\bigl(\funC{1}+\funD{1}\bigr) - \funPts\bigl(\funD{1}\bigr) \Bigr] \label{eq:confmodel:areasize:H} \\
  & \qquad \quad + \sum_{n=2}^\infty w^n \Bigl[ \funPts\bigl(\funC{n}+\funD{n}\bigr) - \funPts\bigl(\funD{n}\bigr) \Bigr] \biggr\} .
  \end{split}
\end{align}%
\label{eq:confmodel:areasize}%
\end{subequations}%

The $\funD{n}(w,x)$ [Eqs.~\eqref{eq:confmodel:areasize:D1}--\eqref{eq:confmodel:areasize:Dn}, pattern~\ref{pattern:nonrootatcapacityreceivesnormal}] are obtained very similarly to the $\funB{}{n}$ [Eq.~\eqref{eq:threereg:areasize:B}]. In the case of $\funD{1}(w,x)$, there is no particular constraint to be satisfied, so the children of the considered pattern contribute a total of $\funQt\bigl( \funC{1} + \funD{1} \bigr)$ (the children not at capacity, and hence not in $\widetilde{\mathcal{G}}$, each contribute a factor $1$). In the case of $\funD{n}(w,x)$ for $n>1$ [Eq.~\eqref{eq:confmodel:areasize:Dn}], we know that the children of the considered node all topple at least once, so they were all at capacity (i.e., they are all in $\widetilde{\mathcal{G}}$) and thus contribute a total of $\funQts\bigl( \funC{n} + \funD{n} \bigr)$.

Obtaining $\funC{1}(w,x)$ [Eq.~\eqref{eq:confmodel:areasize:C1}] is very similar to obtaining $\funA{}{1}(w,x)$ [Eq.~\eqref{eq:threereg:areasize:A1}], except that the patterns~\ref{pattern:nonrootnotatcapacityreceives}--\ref{pattern:nonrootnotatcapacitydissipates} need not be considered because the considered node is know to be at capacity. Similarly, the case of $\funC{2}(w,x)$ and $\funC{n>2}(w,x)$ [Eqs.~\eqref{eq:confmodel:areasize:C2}--\eqref{eq:confmodel:areasize:Cn}] is similar to the one of $\funA{}{n>1}(w,x)$ [Eq.~\eqref{eq:threereg:areasize:An}]. For $\funC{2}(w,x)$, only the contribution of the children in the last term [corresponding to pattern~\ref{pattern:nonrootatcapacityreceivesspecial}] requires further explanations. Using $\funQt\bigl( \funC{1} + \! \funD{1} \bigr)$ would consider that children may have signature $(1,0)$ or $(1,1)$, which includes a spurious term where all children have signature $(1,1)$; subtracting $\funQts\bigl( \funD{1} \bigr)$ cancels this spurious term. A similar cancellation is used in $\funC{n>2}(w,x)$, but this time all the children are known to be at capacity (since they topple at least once).

Finally, $\funH(w,x)$ [Eq.~\eqref{eq:confmodel:areasize:H}] is obtained very similarly to its counterpart for $3$-regular networks [Eq.~\eqref{eq:threereg:areasize:H}]. With probability $1-\sigma$, the root is not at capacity [pattern~\ref{pattern:rootnotatcapacity}]. With probability $\sigma$, the root is at capacity [pattern~\ref{pattern:rootatcapacity}]: a cancellation similar to the one of $\funC{2}(w,x)$ [resp.\ $\funC{n>2}(w,x)$] is used if the root topples once (resp.\ more than once).

One may recover Eq.~\eqref{eq:confmodel:standard} by setting $w = 1$ [with the correspondences $G(x) = \funC{1}(1,x) + \funD{1}(1,x)$ and $H(x) = H(1,x)$, using the ansatz $\funD{n}(1,x) = \funC{n+1}(1,x) + \funD{n+1}(1,x)$]. Alternatively, one may recover Eq.~\eqref{eq:threereg:areasize} by setting $\sigma = \psi_2$, $\funPt(\xi) = [(1-\phi_{22})+\phi_{22}\xi]^3$, $\funQt(\xi) = [(1-\phi_{22})+\phi_{22}\xi]^2$, $\funPts(\xi) = (\phi_{22}\xi)^3$, and $\funQts(\xi) = (\phi_{22}\xi)^2$, with the correspondences $\funA{}{1}(w,x) = (1-\phi_{22})+\funC{1}(w,x)\phi_{22}$, $\funA{}{n>1}(w,x) = \funC{n>1}(w,x)\phi_{22}$ and $\funB{}{n}(w,x) = \funD{n}(w,x)\phi_{22}$; the PGF $\funH(w,x)$ remains $\funH(w,x)$. 

It has been mentioned earlier that one could improve the model presented in Sec.~\ref{section:threereg} by accounting for $3$-star correlations. A starting point for doing so could be the following special case of the configuration model considered in Eq.~\eqref{eq:confmodel:areasize} obtained by fixing $\sigma = \psi_2$ and
\begin{subequations}%
\begin{align}%
\funPt(\xi) & = \theta_0 + \theta_1 \xi + \theta_2 \xi^2 \! + \theta_3 \xi^3 \! &
\funPts(\xi) & = \theta_3 \xi^3 \\
\funQt(\xi) & = \frac{\theta_1 + 2\theta_2 \xi + 3\theta_3 \xi^2}{\theta_1 + 2\theta_2 + 3\theta_3} &
\funQts(\xi) & = \frac{3\theta_3 \xi^2}{\theta_1 + 2\theta_2 + 3\theta_3}.
\end{align}\label{eq:confmodel:threeregwiththreestar}%
\end{subequations}%
Directly using the values of Eq.~\eqref{eq:motivation:oneoverkisbad:theta:measured} would cause the same problems as described at the end of Sec.~\ref{subsection:threereg:areasize}. In particular, the predicted branching factor $R_0 = (2\theta_2+6\theta_3)/(\theta_1+2\theta_2+3\theta_3) \approx 1.03$ is still supercritical. We expect that accounting for higher order correlations should gradually decrease the predicted branching factor towards unity~\footnote{
%
% THE FOOTNOTE
Indeed, the sand in large clusters of nodes at capacity will likely disperse due to large cascades occurring on those nodes. Hence, we expect the equilibrium state to contain smaller clusters of nodes at capacity than one would expect in a randomized surrogate.
}.

More generally for configuration model networks, one will face the same kind of problems when feeding values obtained from simulations into Eq.~\eqref{eq:confmodel:areasize}. To circumvent this problem, one could continue to generalize the method of Sec.~\ref{section:threereg} by implementing a system similar to Eq.~\eqref{eq:threereg:isandnodesiisandlinks}, this time for configuration model networks, and then perform an analysis similar to Sec.~\ref{subsection:threereg:bootstrap} to obtain the equilibrium state without the need for simulations. However, considering the general case may prove tedious, and specifically considering special cases of interest [e.g., Eq.~\eqref{eq:confmodel:threeregwiththreestar}] is likely more promising. We defer to future work the consideration of these questions.
%
%****************************************************************
\section{Conclusion \label{section:conclusion}}
The BTW sandpile process is an archetypal example of self-organized criticality (SOC), a term blanketing any process that has a critical point as an attractor. We argue that, although this article focuses on the BTW sandpile process on networks, the understanding and tools developed herein are applicable to the modeling of a vast array of SOC processes.

To predict the ``macroscopic observables'' of an SOC process (e.g., the exponent of the power-law tail of the cascade size distribution in the BTW process), it may suffice to build a model that reproduces the ``symmetries'' of the SOC process [e.g., the branch distribution (the probability distribution for the number of events caused by an event) should have mean one and the right tail behavior]. This approach will succeed if the resulting model falls in the same universality class as the studied SOC process; other details may not matter. We conjecture in Sec.~\ref{subsection:motivation:oneoverkisbad} and Appendix~\ref{appendix:universality} that sharing a universality class with the BTW process on networks explains the success of an earlier model~\cite{Goh2003} despite its flawed $1/k$ assumption.

In other contexts, a macroscopic understanding may not suffice because, for example, we seek quantities of a microscopic nature, and/or because the process is not critical. Instead, a microscopic understanding of the process is likely required. In the context of the BTW process, we studied the internal workings of a cascade (Sec.~\ref{section:microscopic}), which enabled the possibility to distinguish cascade size and area (Sec.~\ref{subsection:threereg:areasize} and Sec.~\ref{subsection:confmodel:areasize}).

As discussed in Sec.~\ref{subsection:motivation:selfconsistency}, it is important to acknowledge that a model of an SOC process is a different dynamical system in itself: forcing the parameters of a model to those measured in the SOC process may not be the best approach. The alternative that we recommend is to instead force the model to be self-consistent, with the hope that the model self-organizes at the right equilibrium state. Sections~\ref{subsection:threereg:changesYZ} and~\ref{subsection:threereg:bootstrap} perform this feat in the context of the BTW process on random $3$-regular networks: the resulting model predicts with appreciable accuracy nontrivial pairwise correlations in the equilibrium state, and it allows us to explore the BTW process in ranges prohibitively costly to simulate. Section~\ref{subsection:confmodel:areasize} explains how one could generalize the approach to more complex networks.

Furthermore, our method can allow one to study the effects of controlling an SOC process~\cite{Noel2013short}. Though our results were obtained in a specific context, this work serves as a proof of concept that paves the road for powerful, self-consistent and microscopically accurate models of real world systems that self-organize to critical points.

Finally, in the broader context of modeling processes taking place on random graphs, we emphasize that our use of multitype branching processes effectively allows repeated ``back-and-forth'' exchanges of information between a parent vertex and its children. In fact, the signature of a pattern associated with a vertex $v$ provide to $v$'s ancestors, to $v$'s siblings, and to the descendants of $v$'s siblings important information concerning $v$'s descendants, and vice versa~%
\footnote{To give an example, suppose that $v$ has signature $(3,2)$, i.e., $v$ receives $3$ grains from its parent, who receives $2$ grains from $v$. Through that signature, $v$'s parent effectively learns that $v$ topples at least $2$ times, which implies that $v$'s children topple at least $1$ time. Moreover, $v$ learns that all of $v$'s ancestors topple at least $3$ times, that $v$'s siblings topple at least $2$ times, and that the children of $v$'s siblings all topple at least $1$ time.}%
. We feel that this perspective is currently underused in the literature, despite its great potential to analytically model systems in which consequences spread in a bidirectional fashion.

%
%****************************************************************
\begin{acknowledgments}
The authors thank Kwang-Il Goh for useful discussion. This work was supported in part by the Defense Threat Reduction Agency Basic Grant No. HDTRA1-10-1-0088; the Army Research Laboratory Cooperative Agreement W911NF-09-2-0053; the Department of Defense (DoD) MURI award 63826-NS-MUR; the Department of Defense (DoD) through the National Defense Science \& Engineering Graduate Fellowship (NDSEG) Program (C. D. B.); and the Fonds de recherche du Qu\'ebec--Nature et technologies (FRQNT) (P.-A. N.).
\end{acknowledgments}
%
%****************************************************************
\appendix
%****************************************************************
\section{Justification of the universality conjecture \label{appendix:universality}}

This section provides informal justifications to the conjecture made in Sec.~\ref{subsection:motivation:oneoverkisbad}, namely that the reason why the relation $\tau = \gamma/(\gamma-1)$ for scale-free graphs with degree exponent $2 < \gamma < 3$~\cite{Goh2003, Lee2004, Goh2005} is verified by numerical simulations \emph{despite} the known flaws of the $1/k$-assumption is because the branching process leading to that relation belongs to the same universality class as the BTW process on networks.

We first consider a generic Galton-Watson branching process starting with a single particle. The process is completely defined by the sequence $(u_0,u_1,u_2,\ldots)$ such that $u_k$ is the probability that a particle ``disintegrates'' into $k$ new identical particles. We seek the sequence $(v_0,v_1,v_2,\ldots)$ such that $v_s$ is the probability for the total number of particles ever existing to be $s$.

We define the PGFs
\begin{equation}
  \funU(\xi) = \sum_{k=0}^\infty u_k \xi^k
  \quad \text{and} \quad
  \funV(w) = \sum_{s=0}^\infty v_s w^s . \label{eq:universality:branchingUV}
\end{equation}
The normalization of the probabilities $u_k$ forces $\funU(1) = 1$. A well known result is that if $\sum_k k u_k = 1$ [i.e., $R_0 = 1$] then:
(i) the fixed point of $\funV(w) = w \, \funU\biglb( \funV(w) \bigrb)$ determines $\funV(w)$;
(ii) the distribution $v_s$ is normalized (i.e., $\sum_{s=0}^\infty v_s = \funV(1) = 1$); and
(iii) $v_s$ asymptotically follows a power-law $v_s \sim s^{-\tau}$ (possibly with some corrections).
Note that the condition $\sum_k k u_k = 1$ may be understood as $\lim_{\xi \to 1^-} \funU'(\xi) = 1$, where the prime denotes differentiation, and where the limit is taken from the left along the real axis. The case where $u_k$ is ``light tailed'', i.e., $\funU(\xi)$ is analytic in the complex neighborhood of $\xi = 1$, leads to the ``mean field'' regime $\tau = 3/2$.

Goh et al.~\cite{Goh2003} studied the case of the ``heavy tailed'' distribution $u_k^{\text{Goh}} = \alpha k^{-\gamma}$ for $k > 0$ with $\alpha$ and $u_0^{\text{Goh}}$ chosen such that $\funU_{\text{Goh}}(1) = 1$ and $\lim_{\xi \to 1^-} \funU_{\text{Goh}}'(\xi) = 1$, i.e.,
\begin{equation}
  \funU_{\text{Goh}}(\xi) = 1 + \frac{\Li_\gamma(\xi) - \zeta(\gamma)}{\zeta(\gamma - 1)} .
\end{equation}
In this case, they showed that
\begin{equation}
  v_s^{\text{Goh}} \sim
  \begin{cases}
    s^{-\gamma/(\gamma-1)} & (2 < \gamma < 3) \\
    s^{-3/2} (\ln s)^{-1/2} & (\gamma = 3) \\
    s^{-3/2} & (\gamma > 3)
  \end{cases}
\end{equation}
(note that the distribution $u_k^{\text{Goh}}$ cannot be normalized for $\gamma \le 2$). In their demonstration, Goh et al. used the fact that the behavior of $v_s$ in the asymptotic limit $s \to \infty$ can be obtained from the singular behavior of $\funU(\xi)$ near $\xi = 1$. 

We now generalize this result. Define $\funU_{\text{Ana.}}(\xi)$ as an analytic function in the neighborhood of $\xi = 1$ respecting $\funU_{\text{Ana.}}(1) = 1$ and $\funU_{\text{Ana.}}'(1) = 1$. Also define the (possibly singular) function $\funU_{\text{Sing.}}(\xi)$ such that $\funU_{\text{Sing.}}(1) = 1$ and $\lim_{\xi \to 1^-} \funU_{\text{Sing.}}'(\xi) = 1$. For $\rho$ respecting $0 < \rho \le 1$, we see that $\widetilde{\funU}(\xi) = (1-\rho) \funU_{\text{Ana.}}(\xi) + \rho \funU_{\text{Sing.}}(\xi)$ respects $\widetilde{\funU}(1) = 1$, $\lim_{\xi \to 1^-} \widetilde{\funU}'(\xi) = 1$, and $\widetilde{\funU}(z)$ has the same singular behavior in the neighborhood of $\xi = 1$ as $\funU_{\text{Sing.}}(\xi)$. Moreover, the sequence $\widetilde{u}_k = (1-\rho) u_k^{\text{Ana.}} + \rho u_k^{\text{Sing.}}$ corresponding to this $\funU_{\text{Sing.}}(\xi)$ is a probability distribution (i.e., $0 \le \widetilde{u}_k \le 1$ for all $k \ge 0$). Hence, the branching process Eq.~\eqref{eq:universality:branchingUV} using the branching PGF $\funU_{\text{Sing.}}(\xi)$ results in $\widetilde{v}_s$ that have the same behavior in the asymptotic limit $s \to \infty$ as one would obtain with any other choice of $0 < \rho \le 1$ and/or $\funU_{\text{Ana.}}(\xi)$ respecting the aforementioned conditions.

In particular, choosing $\funU_{\text{Sing.}}(\xi) = \funU_{\text{Goh}}(\xi)$, the $\widetilde{u}_k$ obtained with such an ``analytic perturbation'' leads to the the asymptotic behavior described by Goh et al. We \emph{do not} claim that the model of Goh et al. amounts to an ``analytic perturbation'' of the actual BTW process on networks. We only want to point out that there is a continuum of models that belong to the same universality class as the model of Goh et al. for fixed $\gamma > 2$, and that it is not impossible that the BTW process on networks is one of them.
%
%****************************************************************
\section{Proofs for Sec.~\ref{subsection:microscopic:timesthatnodescantopple} \label{appendix:proofs:timesthatnodescantopple}}

This section proves statements from Sec.~\ref{subsection:microscopic:timesthatnodescantopple}, namely Theorems~\ref{theorem:constraintsroot}--\ref{theorem:constraintsnonroot} and their associated Corollary~\ref{corollary:rulespatterns}.
%
%********
Before proving Theorem~\ref{theorem:constraintsroot}, we first consider the following lemma.
\begin{lemma}[Constraints on any cascade] \label{lemma:anycascade}
In any cascade, the following statements hold for any positive integer $n$:
\begin{enumerate}
\item A non-root node cannot topple for the $n$th time before it receives $n$ grains from a single neighbor. \label{lemma:anycascade:nonrootmustreceivenbeforetoppling}
\item No node may receive $n$ grains from a single neighbor before the root topples an $n$th time. \label{lemma:anycascade:rootfirst}
\end{enumerate}
\end{lemma}
\begin{proof}
To prove \ref{lemma:anycascade:nonrootmustreceivenbeforetoppling}, first consider the case that a non-root node $v$ has degree $k$, begins with $k-1$ grains before the cascade, and at some time $t$ in the cascade has received exactly $n-1$ grains from each of its $k$ neighbors. Then $v$ has $k n -1$ grains initially on it and sent to it by time $t$, so $v$ has toppled at most $n-1$ times by time $t$. By construction, $v$ must receive an $n$th grain from at least one neighbor before $v$ can topple an $n$th time. If $v$ received fewer than $n-1$ grains from one or more of its neighbors by time $t$, then $v$ still cannot topple an $n$th time until it receives an $n$th grain from a single neighbor. This argument proves \ref{lemma:anycascade:nonrootmustreceivenbeforetoppling}.

To prove \ref{lemma:anycascade:rootfirst}, let $t$ be the first time in a cascade at which a non-root node receives an $n$th grain from a single neighbor. Let $v$ be one such node. Then a neighbor $u$ of $v$ toppled an $n$th time before time $t$. Suppose for contradiction that $u$ is not the root. In order to topple $n$ times before time $t$, node $u$ must have received at least $n$ grains from one of its neighbors by time $t$, which contradicts the definition of $t$. Thus $u$ must be the root, and claim \ref{lemma:anycascade:rootfirst} follows. 
\end{proof}

%********
We use Lemma~\ref{lemma:anycascade} to prove Theorem~\ref{theorem:constraintsroot} of Sec.~\ref{subsection:microscopic:timesthatnodescantopple}.
\begin{proof}[Proof of Theorem~\ref{theorem:constraintsroot}.]
To show the first claim, note that if the root is at capacity then it topples at least once, and if the root is not at capacity then it does not topple and the cascade ends there. 

To prove the rest, fix $n\geq 1$. If (a) and (b) hold, then the number of grains initially on and received by the root (including the first grain dropped on it) by time $t$ is in the interval $[kn, k(n+1)-1]$, so the root topples $n$ times by time $t$.

Inversely, if (a) does not hold, then the root topples $0<n$ times. It remains only to show that if (b) does not hold then the root does not topple $n$ times by time $t$. There are three cases. First, if the root has received $n$ grains from each of its neighbors by time $t$, then (by counting grains) we know that the root has toppled $n+1$ times by time $t$. Second, if the root has received $m>n$ many grains from a neighbor by time $t$, then that neighbor (call it $u$) must have toppled at least $m$ times by time $t$. Thus, $u$ must have received at least $m$ grains from at least one of its neighbors by time $t$. But by Lemma~\ref{lemma:anycascade}\ref{lemma:anycascade:rootfirst}, no node can receive $m$ grains from a single neighbor before the root topples $m$ times. Hence the root toppled at least $m>n$ times by time $t$. 

In the third and final case, the root has received $m<n-1$ grains from at least one neighbor by time $t$. Let $t'$ be the time when the root topples for the $n$th time. Before $t'$, no neighbor of the root can have received $n$ grains from a single neighbor by Lemma~\ref{lemma:anycascade}\ref{lemma:anycascade:rootfirst} because the root has toppled $\leq n-1$ times. Hence, by Lemma~\ref{lemma:anycascade}\ref{lemma:anycascade:nonrootmustreceivenbeforetoppling}, no neighbor of the root can have toppled $n$ times before time $t'$. Thus, the root cannot have received $\geq n$ grains from a single neighbor by time $t'$. To conclude, the number of grains initially on and received by the root (including the first grain dropped on it) by time $t'$ is $\leq kn+1+(m - n) < kn$, which contradicts to the root toppling for the $n$th time at time $t'$. This concludes the proof.
\end{proof}

%********
Lemma~\ref{lemma:Gprimetree} facilitates the proof of Theorem~\ref{theorem:constraintsnonroot}.
\begin{lemma}[Constraints on cascades that form a finite tree] \label{lemma:Gprimetree}
The following statements hold for a cascade that forms a finite tree $\mathcal{G}^\dagger$.
\begin{enumerate}
\item Let $n$ be a non-negative integer and $v$ be any node in $\mathcal{G}^\dagger$. No descendant of $v$ may receive an $n$th grain from a single neighbor before $v$ topples an $n$th time. \label{lemma:Gprimetree:ancestorfirst}
\item Let $n$ be a non-negative integer and $v$ be any non-root node in $\mathcal{G}^\dagger$. Then node $v$ cannot topple an $n$th time before receiving an $n$th grain from its parent. \label{lemma:Gprimetree:receivenfromparentbeforetopplen}
\item Let $n$ be a positive integer, and let $v$ be a non-root node that has received $n$ grains from its parent by time $t$. Then node $v$ has toppled $n$ times by time $t$ if and only if it is initially at capacity, and at the moment of toppling for the $n$th time it has received $n-1$ grains from every one of its children and $n$ grains from its parent. \label{lemma:Gprimetree:nonrootchildrennminusone}
\item Let $n$ be a positive integer, and let $v$ be a non-root node that has toppled $n$ times by time $t$. Then $v$'s parent toppled at most $n+1$ times by time $t$. \label{lemma:Gprimetree:nonrootmaxfromparent}
\item Let $v$ be any node in $\mathcal{G}^\dagger$. If $v$ is not at capacity, then $v$ topples $0$ times. \label{lemma:Gprimetree:notcapacity}
\end{enumerate}
\end{lemma}
\begin{proof}
To show~\ref{lemma:Gprimetree:ancestorfirst}, first note that if $v$ is the root, then the claim follows directly from the analogous (but weaker) result, Lemma~\ref{lemma:anycascade}\ref{lemma:anycascade:rootfirst}, for cascades that do not form trees. To finish proving~\ref{lemma:Gprimetree:ancestorfirst}, suppose $v$ is not the root, and assume (for contradiction) that a descendant of $v$ receives an $n$th grain from a single neighbor before $v$ topples an $n$th time. Let $t$ be the first time when a descendant of $v$ receives $n$ grains from a single neighbor; call such a descendant $u$, and let $w$ be a neighbor of $u$ from which $u$ receives an $n$th grain at time $t$. Then $w$ must have toppled for an $n$th time at time $t' < t$. We know that $w$ is not $v$ because $v$ does not topple for the $n$th time before time $t$. Because $w$ is a descendant of $v$, we know $w$ is not the root, so by Lemma~\ref{lemma:anycascade}\ref{lemma:anycascade:nonrootmustreceivenbeforetoppling} we know that $w$ must have received at least $n$ grains from a single neighbor before time $t'$, which contradicts the definition of $t$. Thus, claim~\ref{lemma:Gprimetree:ancestorfirst} follows.

We show \ref{lemma:Gprimetree:receivenfromparentbeforetopplen} by contradiction. Suppose (for contradiction) that $v$ does not receive an $n$th grain from its parent by time $t$ and yet $v$ topples an $n$th time at time $t$. Before time $t$, $v$ has toppled fewer than $n$ times, so (by Lemma~\ref{lemma:Gprimetree}\ref{lemma:Gprimetree:ancestorfirst} applied to $v$) no children of $v$ have received $n$ grains from the same neighbor. Thus, no children of $v$ have toppled $n$ times by time $t$ [by Lemma~\ref{lemma:anycascade}\ref{lemma:anycascade:nonrootmustreceivenbeforetoppling}], so $v$ does not receive $n$ grains from the same child by time $t$. But by Lemma~\ref{lemma:anycascade:nonrootmustreceivenbeforetoppling} and the assumption that $v$ topples at time $t$, we know $v$ must have received $n$ grains from a single neighbor before time $t$, a contradiction. Thus, claim~\ref{lemma:Gprimetree:receivenfromparentbeforetopplen} follows. 

Claim \ref{lemma:Gprimetree:nonrootchildrennminusone} follows from counting grains of sand. Let $k$ be the degree of the non-root node $v$ that has received $n \ge 1$ grains from its parent by time $t$. Suppose that $v$ topples for an $n$th time at some time $t' \le t$. Before time $t'$, Lemma~\ref{lemma:Gprimetree}\ref{lemma:Gprimetree:ancestorfirst} guarantees that the children of $v$ toppled at most $n-1$ times, so $v$ received at most $n-1$ grains from each of its children by time $t'$. Moreover, $v$ received at most $n$ grains from its parent by time $t'$ (because it received $n$ grains from its parent by time $t \ge t'$). Because $v$ topples for the $n$th time at time $t'$, the total number of grains initially on and received by $v$ by time $t'$ should be $kn$. From the preceding constraints, this is only possible if $v$ is initially at capacity, $v$ receives $n-1$ grains from each of its children by time $t'$, and $v$ receives $n$ grains from its parent by time $t'$.

Conversely, suppose that $v$ is initially at capacity, and suppose that $t'$ is the first time such that $v$ has received $n-1$ grains from every one of its children and $n$ grains from its parent. Then, by the previous grain-counting argument, $v$ topples an $n$th time at time $t'$. It remains to be proven that $v$ does not topple again by time $t$, which is guaranteed by Lemma~\ref{lemma:Gprimetree}\ref{lemma:Gprimetree:receivenfromparentbeforetopplen} because $v$ received only $n$ grains from its parent by time $t$. Thus claim~\ref{lemma:Gprimetree:nonrootchildrennminusone} holds.

We show\ \ref{lemma:Gprimetree:nonrootmaxfromparent} by induction over the generation $g \ge 1$ of the non-root $v$ that topples $n \ge 1$ times by time $t$. Let $u$ be the parent of $v$. Suppose $u$ is the root (i.e., $g = 1$), and suppose (for contradiction) that $u$ topples $m \geq n+2$ times by time $t$. By Theorem~\ref{theorem:constraintsroot}, we know that $u$ has received either $m$ or $m-1$ grains from each of its children by time $t$, including node $v$. However, by assumption, $v$ topples $n \le m-2$ times by time $t$, so $u$ receives $\leq m-2$ grains from $v$ by time $t$, a contradiction. Thus, claim\ \ref{lemma:Gprimetree:nonrootmaxfromparent} holds for $g=1$.

Now suppose a node $v$ at generation $g>1$ topples $n \geq 1$ times by time $t$, and assume that the claim holds at generation $g-1$. Suppose (for contradiction) that $v$'s parent, $u$, toppled $m \geq n+2$ times by time $t$, and let $t' \le t$ be the moment when $u$ topples for the $m$th time. Before time $t'$, $u$ toppled at most $m-1$ times, so by the inductive hypothesis we know that $u$'s parent has toppled $\leq m$ times. Moreover, by Lemma~\ref{lemma:Gprimetree}\ref{lemma:Gprimetree:receivenfromparentbeforetopplen}, $u$ receives at least $m$ grains from its parent by time $t'$ because its parent topples an $m$th time. Thus, $u$ receives exactly $m$ grains from its parent by time $t'$, and $u$ topples an $m$th time at time $t'$, so we can apply Lemma~\ref{lemma:Gprimetree}\ref{lemma:Gprimetree:nonrootchildrennminusone} to $u$ to conclude that $u$ must have received $m-1$ grains from each of its children (including node $v$) by time $t'$. But $v$ has toppled $n \le m-2$ times by time $t$, so there is no time $t' \le t$ at which $u$ receives an $(m-1)$th grain from $v$, a contradiction. Thus, \ref{lemma:Gprimetree:nonrootmaxfromparent} follows by induction on $g$. 

Claim \ref{lemma:Gprimetree:notcapacity} is already shown if $v$ is the root (Theorem~\ref{theorem:constraintsroot}); here we show the case in which $v$ is non-root by contradiction. If $v$ is not at capacity, then it must receive at least $2$ grains (from any source) before toppling. Consider (for contradiction) the time $t' \le t$ at which $t$ receives a second grain. Before $t'$, $v$ has toppled $0$ times, so $v$ cannot receive grains from its children [by Lemma~\ref{lemma:Gprimetree}\ref{lemma:Gprimetree:ancestorfirst}] and $v$ cannot have received more than one grain from its parent [because $v$'s parent cannot have toppled more than once by Lemma~\ref{lemma:Gprimetree}\ref{lemma:Gprimetree:nonrootmaxfromparent}]. Hence, $v$ cannot receive a second grain by time $t'$, a contradiction. So claim \ref{lemma:Gprimetree:notcapacity} holds.
\end{proof}

%********
We use Lemmas~\ref{lemma:anycascade}--\ref{lemma:Gprimetree} and Theorem~\ref{theorem:constraintsroot} to prove Theorem~\ref{theorem:constraintsnonroot} of Sec.~\ref{subsection:microscopic:timesthatnodescantopple}.
\begin{proof}[Proof of Theorem~\ref{theorem:constraintsnonroot}.]
We first prove the first claim concerning $v$ toppling $0$ times. By Lemma~\ref{lemma:Gprimetree}\ref{lemma:Gprimetree:notcapacity}, $v$ topples $0$ times if $v$ is not initially at capacity. Now suppose that the non-root node $v$ receives $0$ grains from its parent by time $t$. By Lemma~\ref{lemma:Gprimetree}\ref{lemma:Gprimetree:ancestorfirst} with $n=0$, $v$ cannot receive a grain from one of its children by time $t$. Thus, $v$ receives a total of $0$ grains by time $t$ and hence topples $0$ times by time $t$.

Inversely, if $v$ is initially at capacity and $v$ receives at least one grain from its parent by time $t$, then $v$ clearly topples at least once by time $t$. This concludes the proof of the claim for $v$ toppling $0$ times.

To prove the second claim, first suppose that conditions (a), (b) and (c) hold. Let $k$ be the degree of $v$. By counting grains, we see that the number of grains initially on $v$ and received by $v$ by time $t$ is in the interval $[kn, k(n+1)-1]$, so $v$ topples $n$ times by time $t$.

To show the converse, suppose one of (a), (b) or (c) does not hold, and suppose (for contradiction) that $v$ topples $n$ times by time $t$. If (a) does not hold (i.e., if $v$ is not initially at capacity), then $v$ topples $0$ times by time $t$ [Lemma~\ref{lemma:Gprimetree}\ref{lemma:Gprimetree:notcapacity}], a contradiction. 

Next consider the two cases in which (b) does not hold. If $v$ received fewer than $n$ grains from its parent by time $t$, then by Lemma~\ref{lemma:Gprimetree}\ref{lemma:Gprimetree:receivenfromparentbeforetopplen} $v$ topples fewer than $n$ times by time $t$, a contradiction. If $v$ received $>n+1$ grains from its parent by time $t$ and if $v$'s parent were the root, then by Theorem~\ref{theorem:constraintsroot} node $v$ would necessarily topple $\geq n+1$ times by time $t$, a contradiction. Finally, if $v$ received $> n+1$ grains from its parent by time $t$ and if $v$'s parent were not the root, then by part (c) of this theorem applied to the parent of $v$, we know that $v$ topples $\geq n+1$ times by time $t$, a contradiction. 

To finish the proof, consider the three ways in which (c) may not hold. If $v$ received fewer than $n-1$ grains from any of its children by time $t$, then $v$ toppled fewer than $n$ times by time $t$ [by Lemma~\ref{lemma:Gprimetree}\ref{lemma:Gprimetree:nonrootchildrennminusone}], a contradiction. If $v$ received more than $n$ grains from any of its children by time $t$, then $v$ toppled more than $n$ times by time $t$ [because Lemma~\ref{lemma:Gprimetree}\ref{lemma:Gprimetree:receivenfromparentbeforetopplen} implies that such a child must have received more than $n$ grains from its parent, i.e., from $v$], a contradiction. In the last case, $v$ received $n+1$ grains from its parent and $n$ grains from each of its children. Let $k$ be the degree of $v$. By counting grains, we see that the number of grains initially on $v$ and received by $v$ by time $t$ is $k(n+1)$, so $v$ topples $n+1$ times by time $t$, a contradiction. Hence the second claim holds, which completes the proof.
\end{proof}

%********
Finally, we use Theorems~\ref{theorem:constraintsroot}--\ref{theorem:constraintsnonroot} to prove Corollary~\ref{corollary:rulespatterns} of Sec.~\ref{subsection:microscopic:timesthatnodescantopple}.
\begin{proof}[Proof of Corollary~\ref{corollary:rulespatterns}]
We prove claims~\ref{pattern:rootnotatcapacity}--\ref{pattern:nonrootatcapacitydissipatesout} by letting $t$ be some time after the cascade finishes and by applying Theorems~\ref{theorem:constraintsroot}--\ref{theorem:constraintsnonroot}. 

We first consider the cases in which $v$ topples $0$ times. Suppose $v$ is the root. By Theorem~\ref{theorem:constraintsroot}, point~\ref{pattern:rootnotatcapacity} is necessary and sufficient for $v$ to topple $0$ times.

Now suppose $v$ is non-root and topples $0$ times. Because $v$ topples zero times, $v$'s parent receive $0$ grains from $v$, so the parent toppled at most $1$ time (by Theorem~\ref{theorem:constraintsroot} or Theorem~\ref{theorem:constraintsnonroot} if the parent is the root or not, respectively), and at least one time (otherwise $v$ would not be in $\mathcal{G}'$). Thus, the parent of $v$ has toppled exactly one time. Hence the signature of $v$ is $(1,0)$. By Theorem~\ref{theorem:constraintsnonroot}, $v$ is not initially at capacity, and/or $v$ receives $0$ grains from its parent during the cascade. These conditions leave only $3$ possibilities: $v$ is not initially at capacity and receives the grain sent by its parent [point~\ref{pattern:nonrootnotatcapacityreceives}]; $v$ is not initially at capacity and the grain dissipates [point~\ref{pattern:nonrootnotatcapacitydissipates}]; or $v$ is initially at capacity and the grain dissipates [point~\ref{pattern:nonrootatcapacitydissipates}]. There are no other possibilities.

Next consider the cases in which $v$ topples $n \ge 1$ times. If $v$ is the root, then point (v) is equivalent to Theorem~\ref{theorem:constraintsroot} and the fact that $v$'s children are the same as $v$'s neighbors. 

Now suppose $v$ is not the root. Let $(n',m')$ be the signature of $v$; let $m$ be the amount of sand received by $v$ from its parent; and let $l_c$ be the amount of sand received by $v$ from one of its children $c$ [that child thus has signature $(n,l_c)$]. By Theorem~\ref{theorem:constraintsnonroot}, it is necessary and sufficient that: $v$ was initially at capacity, $v$ received $n$ or $n+1$ grains from its parent (so $m \ge n$), and $v$ received from each of its children $n$ or $n-1$ grains (so $n \ge l_c \ge n-1$ for every child $c$ of $v$), except that $v$ cannot receive $n$ grains from all of its children if it received $n+1$ grains from its parent. Moreover, $v$ cannot receive more grains from its parent than the number of times the parent toppled (so $n' \ge m$); $v$'s parent cannot receive more grains from $v$ than the number of times $v$ toppled (so $n \ge m'$); and $v$'s parent must receive at least $n'-1$ grains from $v$ (by Theorem~\ref{theorem:constraintsroot} or Theorem~\ref{theorem:constraintsnonroot} if the parent is the root or not, respectively, so $m' \ge n'-1$). Grain exchanges between $v$ and its parent may thus be summarized as $n' \ge m \ge n \ge m' \ge n'-1$, which leaves $4$ possibilities (i.e., $4$ possible positions of the ``$>$'' symbol): $n' > m = n = m' = n'-1$ [the last grain sent by the parent of $v$ toward $v$ dissipated, point~\ref{pattern:nonrootatcapacitydissipatesin}], $n' = m > n = m' = n'-1$ [no dissipation, $v$ cannot receive $n$ grains from all its children, point~\ref{pattern:nonrootatcapacityreceivesspecial}], $n' = m = n > m' = n'-1$ [the last grain sent by $v$ toward its parent dissipated, point~\ref{pattern:nonrootatcapacitydissipatesout}], and $n' = m = n = m' > n'-1$ [no dissipation, point~\ref{pattern:nonrootatcapacityreceivesnormal}]. There are no other possibilities.
\end{proof}
%
%****************************************************************
\section{Proofs for Sec.~\ref{subsection:microscopic:areapercolation} \label{appendix:proofs:areapercolation}}

This section proves statements from Sec.~\ref{subsection:microscopic:areapercolation}, namely Theorem~\ref{theorem:AAtilde} and its associated Corollary~\ref{corollary:areabond}.
\begin{proof}[Proof of Theorem~\ref{theorem:AAtilde}.]
If $v$ is not initially at capacity, then it does not topple, so $\mathcal{G}^\dagger$ contains only $v$ (a valid tree), $\mathcal{G}'$ is empty and $\mathcal{A} = 0$. By definition, $v$ does not belong to $\widetilde{\mathcal{G}}$ because $v$ is not at capacity, so $\widetilde{\mathcal{G}}'$ is also empty (a valid tree) and $\widetilde{\mathcal{A}} = 0$.

Otherwise, suppose $v$ is at capacity. Then $v$ topples at least once in the cascades on both $\mathcal{G}$ and $\widetilde{\mathcal{G}}$, so $v$ is in all of the graphs $\mathcal{G}^\dagger$, $\mathcal{G}'$, and $\widetilde{\mathcal{G}}'$. By definition, the number of nodes in $\widetilde{\mathcal{G}}'$ equals the area $\widetilde{\mathcal{A}}$ of the cascade on $\widetilde{\mathcal{G}}$. Since $\widetilde{\mathcal{A}} < M$ and $(\mathcal{G},M,v)$ is good, we know that $(\mathcal{G},\widetilde{\mathcal{A}},v)$ is good, too. Also $(\widetilde{\mathcal{G}},\widetilde{\mathcal{A}},v)$ and $(\widetilde{\mathcal{G}}',\widetilde{\mathcal{A}},v)$ are both good, because $\widetilde{\mathcal{G}}'$ is a subgraph of $\widetilde{\mathcal{G}}$ and because $\widetilde{\mathcal{G}}$ is a subgraph of $\mathcal{G}$. Since $\widetilde{\mathcal{A}} < M$ and $(\widetilde{\mathcal{G}}', \widetilde{\mathcal{A}},v)$ is good, we know that $\widetilde{\mathcal{G}}'$ is a forest. Moreover, because the cascade spreads between adjacent nodes starting from the root $v$, $\widetilde{\mathcal{G}}'$ is connected, and thus $\widetilde{\mathcal{G}}'$ is a tree. 

It remains only to show that the cascade forms a finite tree $\mathcal{G}^\dagger$, because that conclusion will imply the equalities $\mathcal{G}' = \widetilde{\mathcal{G}}'$ and $\mathcal{A} = \widetilde{\mathcal{A}}$. First we show that $\mathcal{G}'$ is a subgraph of $\mathcal{G}^\dagger$. Observe that $\mathcal{G}^\dagger$ can be obtained by the union of the root $v$, the graph $\mathcal{G}'$, the neighbors of nodes that topple, and the links joining nodes that topple to their neighbors. Since there are no links between nodes that do not topple, we may consider individually each additional node $u$ adjacent to a node that topples such that $u$ itself does not topple. Fix one such $u$. We define $\widetilde{\mathcal{G}}''(u)$ to be the subgraph of $\mathcal{G}$ induced by $\{u\} \cup \{\text{all nodes of } \widetilde{\mathcal{G}}' \}$. Now $\widetilde{\mathcal{G}}''(u)$ contains $\widetilde{\mathcal{A}} + 1 \le M$ nodes, so $(\widetilde{\mathcal{G}}''(u),\widetilde{\mathcal{A}}+1,v)$ is good. Thus, $u$ is a leaf of the tree $\widetilde{\mathcal{G}}'$. Hence, the presence of node $u$ in $\mathcal{G}^\dagger$ (together with its single link to a node that topples) cannot introduce a cycle in $\mathcal{G}^\dagger$, so $\mathcal{G}^\dagger$ is a tree and the cascade forms a finite tree $\mathcal{G}^\dagger$. Finally, Theorem~\ref{theorem:constraintsnonroot} guarantees that nodes not at capacity do not topple, so $\mathcal{G}' = \widetilde{\mathcal{G}}'$, and thus $\mathcal{A} = \widetilde{\mathcal{A}}$.
\end{proof}

\begin{proof}[Proof of Corollary~\ref{corollary:areabond}.]
The case $x=0$ is trivial: $\mathcal{A} = 0$ if and only if $v$ is not at capacity, which is true if and only if $v$ belongs to a component of size $0$ in $\widetilde{\mathcal{G}}^{(1-\epsilon)}$.

Otherwise suppose $v$ is at capacity. Then $v$ topples at least once. Let $\widetilde{\mathcal{C}}_v^{\kern2pt (1-\epsilon)}$ denote the connected component to which $v$ belongs in $\widetilde{\mathcal{G}}^{(1-\epsilon)}$. If the number of nodes in $\widetilde{\mathcal{C}}_v^{\kern2pt (1-\epsilon)}$ is $< M$, then because $(\mathcal{G},M,v)$ is good, we know that $\widetilde{\mathcal{C}}_v^{\kern2pt (1-\epsilon)}$ is a tree. Moreover, by Theorem~\ref{theorem:AAtilde}, for $x<M$ we have $\mathcal{G}' = \widetilde{\mathcal{G}}'$ and $\mathcal{A} = \widetilde{\mathcal{A}}$. Thus, for $x<M$, $\widetilde{\mathcal{C}}_v^{\kern2pt (1-\epsilon)}$ is built in the same way as the cascade grows on the subgraph $\widetilde{\mathcal{G}}$ of nodes at capacity (with nodes not at capacity treated as sinks). Specifically, because $x<M$ and $(\mathcal{G},M,v)$ is good, both $\widetilde{\mathcal{C}}_v^{\kern2pt (1-\epsilon)}$ and the subgraph $\mathcal{G}'$ of $\mathcal{G}$ induced by nodes that topple are built by adjoining at-capacity neighbors of leaves independently with probability $1-\epsilon$. Thus, for $x < M$, the chance that $\widetilde{\mathcal{C}}_v^{\kern2pt (1-\epsilon)}$ has $x$ many nodes equals the chance that the area $\mathcal{A}=x$. 
\end{proof}
%
%****************************************************************
%\bibliography{../btw_network_bibliography}
%****************************************************************
%Merlin.mbs v4.21 2009-07-09.
%
%
%****************************************************************
\end{document}